\documentclass{article}
\usepackage{graphicx} 
\usepackage[table,xcdraw]{xcolor}
\usepackage{lmodern} 
\usepackage{float}
\usepackage[T1]{fontenc} 
\usepackage{changepage} 
\usepackage[margin=3.85cm]{geometry}
\usepackage[backend=bibtex,doi=false,isbn=false,url=false,maxbibnames=5,sorting=none]{biblatex}
\AtEveryBibitem{\clearfield{month}}
\AtEveryBibitem{\clearfield{day}}
\usepackage{mathtools} 
\usepackage{amsfonts}
\usepackage{amsmath} 
\usepackage{amsthm}
\usepackage{dsfont}
\usepackage{physics}
\usepackage[ruled]{algorithm2e}
\usepackage{enumerate}
\usepackage[font=small,labelfont=bf]{caption}
\usepackage{tikz}
\usetikzlibrary{hobby} 
\usepackage{subcaption}
\usepackage{graphicx}
\usepackage[normalem]{ulem} 
\usepackage{array} 
\usepackage{color}
\usepackage{todonotes}
\usepackage{cancel}
\usepackage{multirow}

\usepackage{changepage} 
\usepackage{multirow}
\newcolumntype{C}[1]{>{\centering\let\newline\\\arraybackslash\hspace{0pt}}m{#1}}
\newcolumntype{L}[1]{>{\raggedright\let\newline\\\arraybackslash\hspace{0pt}}m{#1}}
\usepackage{hyperref}
\usepackage[capitalise]{cleveref}
\makeatletter
\newtheorem*{rep@theorem}{\rep@title}
\newcommand{\newreptheorem}[2]{\newenvironment{rep#1}[1]{ \def\rep@title{#2 \ref{##1}} \begin{rep@theorem}}{\end{rep@theorem}}}
\makeatother

\newtheorem{theorem}{Theorem}
\newreptheorem{theorem}{Theorem}
\newreptheorem{corollary}{Corollary}
\newtheorem{lemma}[theorem]{Lemma}
\newtheorem{corollary}[theorem]{Corollary}
\newtheorem{definition}[theorem]{Definition}

\newtheorem{problem}[theorem]{Open Problem}
\crefname{section}{Section}{Sections}
\crefname{subsection}{subsection}{subsections}
\crefname{theorem}{Theorem}{Theorems}
\crefname{corollary}{Corollary}{Corollaries}
\crefname{lemma}{Lemma}{Lemmas}
\crefname{appendix}{Appendix}{Appendices}
\crefname{definition}{Definition}{Definitions}
\crefname{equation}{eq.}{eqs.}
\crefname{algorithm}{Algorithm}{Algorithms}
\definecolor{myred}{RGB}{221,0,0}
\definecolor{myblue}{RGB}{0,0,221}

\definecolor{citegreen}{RGB}{0,165,0}
	
\renewcommand{\ket}[1]{| #1 \rangle}

\DeclarePairedDelimiter{\floor}{\lfloor}{\rfloor}
\DeclarePairedDelimiter\ceil{\lceil}{\rceil}

\DeclareMathOperator{\Ima}{Im}
\DeclareMathOperator{\poly}{\mathrm{poly}}
\newcommand{\hbt}{\mathcal{H}}

\newcommand{\swap}{\text{SWAP}}
\newcommand{\idty}{\mathds{1}}
\newcommand{\ee}{\mathrm{e}}
\newcommand{\ii}{\mathrm{i}}

\title{Simulating LDPC code Hamiltonians on 2D lattices}
\author{Harriet Apel\textsuperscript{a}, Nou\'edyn Baspin\textsuperscript{b}}
\date{\small\textit{\textsuperscript{a}Department of Computer Science, University College London, UK}\\  \textit{\textsuperscript{b}
School of Physics, The University of Sydney, Sydney, New South Wales 2006, Australia}\\ \vspace{-3ex}\normalsize}

\begin{document}
\maketitle

\begin{abstract}
    While LDPC codes have been demonstrated with desirable error correcting properties, this has come at a cost of diverging from the geometrical constraints of many hardware platforms. 
    Viewing codes as the groundspace of a Hamiltonian, we consider engineering a simulation Hamiltonian reproducing some relevant features of the code.
    Techniques from Hamiltonian simulation theory are used to build a simulation of LDPC codes using only 2D nearest-neighbour interactions at the cost of an energy penalty polynomial in the system size.
    We derive guarantees for the simulation that allows us to approximately reproduce the ground state of the code Hamiltonian, approximating a $[[N, \Omega(\sqrt{N}), \Omega(\sqrt{N})]]$ code in 2D.
    The key ingredient is a new constructive tool to simulate an $l$-long interaction between two qubits by a 1D chain of $l$ nearest-neighbour interacting qubits using $\poly( l)$ interaction strengths.
    This is an exponential advantage over the existing gadgets for this routine which facilitates the first $\epsilon$-simulation of \emph{arbitrary sparse} Hamiltonian on $n$ qubits with a Hamiltonian on a 2D lattice of $O(n^2)$ qubits with interaction strengths scaling as $O\left(\mathrm{poly}(n, 1/\epsilon)\right)$.
\end{abstract}

\setcounter{tocdepth}{2}
\tableofcontents

\section{Introduction}

As quantum computers are inherently sensitive to noise, quantum codes have become an important primitive in their engineering.
However, despite their apparent necessity, they are also expected to incur significant overhead.
There is thus a strong interest in understanding the most efficient codes an architecture can operate \cite{Bravyi2009,Bravyi2010}, in the hope of limiting these additional costs.
One of the major results from this line of research is that geometrically local architectures can only be so efficient. For example, most codes can be defined as the groundstate of a sparse Hamiltonian $H = \sum_i h_i$, and if additionally constrained to be a 2D commuting projector Hamiltonian -- the case for a large number of codes known and used today -- obey $kd^2 \in O(n), d \in O(\sqrt{n})$ \cite{Bravyi2009,Bravyi2010}.
Here $k$ is the number of logical qubits, $d$ is the distance, and $n$ is the number of physical qubits used in the $[[n,k,d]]$ quantum code.

On the other hand, if the geometrical restriction is dropped, quantum codes can achieve parameters scaling as $[[n, \Theta(n), \Theta(n)]]$ \cite{panteleev2021asymptotically, Breuckmann}. 
There thus exists a landscape of parameters $k,d$ between 2D commuting projector Hamiltonians, and sparse Hamiltonians in general. 
In this work we contribute to the exploration of this landscape by proving that an arbitrary quantum LDPC code can be simulated by a 2D nearest-neighbour Hamiltonian $\tilde{H} = \sum \tilde{h}_j$ acting on $O(n^2)$ qubits, with an interaction strength $\norm{\tilde{h}_j}$ that scales polynomially with the system size and the reciprocal of the simulation error. 

When codes are expressed as a Hamiltonian groundspace, the computations that run within this subspace will be protected from ambient errors.
Errors translate into excitations of the Hamiltonian and measuring these excitations can be used to correct those errors.
So naturally, operating a quantum code requires the ability to reliably produce these ground states and hence requires our simulation to not only preserve the eigenspectrum but also the \emph{eigenstates}. 
This is in stark contrast with the usual simulation results, which guarantee a small difference in energy \emph{eigenvalues} between the simulator Hamiltonian and original code Hamiltonian -- however, that alone is trivial since all $[[n,k,d]]$ codes are unitarily equivalent. 

Our simulation successfully overcomes these limitations, more precisely we show that:
\begin{reptheorem}{thm LDPC}[Simulating LDPC codes (informal) \cref{thm LDPC}]
    Given a Hamiltonian $H = \sum_i h_i$ corresponding to an arbitrary LDPC code $\mathcal{C}$ acting on register $\mathcal{H}$ of $n$ qubits .
    $\exists$ a 2D nearest-neighbour Hamiltonian, $\tilde{H} = \sum_j \tilde{h}_j$, acting on an augmented register $\mathcal{H}\otimes\mathcal{A}$ of $N\in O(n^2)$ qubits such that,
    \begin{enumerate}
        \item The interaction strength of $\tilde{H}$ scales as $\poly(n,\zeta^{-1}))$.
        
        \item \label{item:2}  States close to the groundstates of $\tilde{H}$ are close to the groundstate of $H$. I.e. If $\rho $ is a low-energy state in $\mathcal{H} \otimes \mathcal{A}$ with $\trace(\rho \tilde{H})\leq E$, then $ \tr(\rho_\mathcal{H} H) \in O(E + \zeta)$. 
        \item \label{item:3} All ground states of $H$ are close to the groundstates of $\tilde{H}$ given the right ancilla state. I.e. Let $\sigma$ be a state in the support of $\mathcal{S}(\mathcal{H})$ then, $\exists$ $\tilde{\sigma} \in \mathcal{S}(\mathcal{H}\otimes \mathcal{A})$, close to $\sigma$ on the original register $\norm{\tr_\mathcal{A}(\tilde{\sigma}) - \sigma}_1 \in O(\zeta)$ where $\abs{\tr(\tilde{\sigma} \tilde{H}) -  \tr(\sigma H) } \leq \zeta.$ 
        \item The low-energy eigenspectrum of $\tilde{H}$ $\zeta$-approximates the eigenspectrum of $H$.
     \end{enumerate}
     for any $\zeta,E \geq0$ where $\mathcal{S}(\mathcal{H})$ denotes the set of states acting on $\mathcal{H}$.
\end{reptheorem}

In error correcting terms, item \ref{item:2} states that low-energy states of $\tilde{H}$ have a small syndrome with respect to the code Hamiltonian $H$.
And conversely, item \ref{item:3} states that all code states in the ground space of $H$ can be found in the ground space of $\tilde{H}$, up to a small controllable error. 
Together, those properties show that the low-energy space of $\tilde{H}$ is precisely the low-syndrome space of $H$, up to a small error.

We are free to tune error parameter $\zeta$ to reduce this error at the expense of increasing the interaction strength\footnote{In the formal statement there are multiple error parameters in the simulation describing error in eigenvalues ($\epsilon$) and error in eigenstate ($\eta$) separately, but we use $\zeta$ -- a proxy for total error -- for simplicity here.}.
Since we use $N \in O(n^2)$ qubits in total, we effectively simulate a $[[N, \Omega(\sqrt{N}), \Omega(\sqrt{N})
]]$ code in 2D.
The fact that the groundspace of LDPC codes can be made accessible to 2D architecture is surprising, as they generally require either long-range interactions, or instantaneous classical communications.

The proof of this result heavily relies on tools from the literature on analogue simulation of quantum systems, specifically perturbation gadgets.
Emerging from Hamiltonian complexity theory \cite{Oliveira2008}, perturbation gadgets are one of a handful of constructive simulation techniques.
The gadgets iteratively morph sections of Hamiltonian interaction graphs at the cost of additional qubits.
This allows restricted Hamiltonians to emulate the physics of a larger range of more complex systems in the low energy regime. 
However, the interaction strengths involved to apply multiple rounds of perturbation theory are in general impractical to implement: a $\zeta$-simulation involving $r$ rounds of perturbation on a $n$-qubit Hamiltonian requires strengths of $O\left(\left(\zeta^{-1}\mathrm{poly}(n)\right)^{6^r}\right)$.
Nevertheless, if the number of rounds of perturbation theory is restricted to $O(1)$ then simulator Hamiltonian weights are `only' polynomially growing with the system size.

LDPC codes are described by \emph{general sparse Hamiltonians} whereby the degree and locality is bounded by a constant, but when planarised the interaction graph edges can be $O(n)$ long. 
These long-range interactions are integral for expansion which is closely linked to good error correcting properties.
Using existing gadgets fitting these $O(n)$-long range interactions onto a lattice required $O(\log n)$ rounds of perturbation theory and hence $\exp(n)$ interaction strengths.
It was unclear whether this energy resource was a necessary requirement for the perturbative simulation.
Our main technical contribution is to develop a new gadget that in a single step simulates a 2-qubit interaction with a chain of $n$, 2-qubit interactions using strengths scaling as $O(\mathrm{poly}(n, \zeta^{-1}))$ - an exponential improvement. 

History state simulations~\cite{aharonov2018hamiltonians, Kohler2020} provide an alternative to perturbative techniques. 
In these works the target Hamiltonian is encoded into a 1D quantum phase estimation circuit before obtaining the simulator Hamiltonian from the Feynman-Kitaev circuit-to-Hamiltonian construction~\cite{Kitaev2002ClassicalAQ}.
\cite{Zhou} recently constructed the first protocol for simulating a general local Hamiltonian with a geometrically local Hamiltonian with polynomial resources using these ideas. 
The simulator Hamiltonian in this construction acts on $\mathrm{poly}(n,\zeta^{-1})$ qubits and therefore the lattice size depends on the simulation error.
While our technique additionally requires the target Hamiltonian to be sparse (each qubit acted on non-trivially by a constant number of Hamiltonian terms), the simulator acts on only $O(n^2)$ qubits independent of the simulation error so the same lattice can be used while taking the error arbitrarily small.
The quadratic ancillas in this work is also a polynomial improvement over the general poly in \cite{Zhou}.

This work was motivated by LDPC codes however, the technique provides a general protocol for constructing a 2D lattice simulator Hamiltonian of arbitrary sparse Hamiltonians using only polynomial interaction strengths and quadratically many ancilla qubits.
Analogue simulation is an important near-term application of quantum computers hence the broader simulation result is of interest independently of code Hamiltonians.
\begin{repcorollary}{cor sparse}[Simulating sparse Hamiltonians (informal)]
    Given a Pauli Hamiltonian, $H = \sum_i h_i$, acting on $n$ qubits that is \emph{sparse}: $O(1)$-local and the maximum degree is $O(1)$.

$\exists$ a nearest-neighbour Hamiltonian, $\tilde{H} = \sum_j \tilde{h}_j$, acting on a 2D lattice of $N = O(n^2)$ qubits, $\tilde{H} = \sum_j \tilde{h}_j$,
that is a simulation of $H$ with interaction strength scaling polynomially in $n$ and inverse polynomially with the precision of the simulation.
\end{repcorollary}

\subsection{Open questions}


The main drawback of our construction is the polynomial interaction strength it incurs.
Having $\norm{h_i}$ depend solely on $\zeta$ instead would be highly desirable and constitute a natural research question: can the dependence be improved, or is it possible to lower bound $\norm{h_i}$?
Unfortunately, we show in Appendix \ref{sec:no-go} that superficial improvements to our gadget cannot yield this scaling.
We first highlight two weakenings of our main result that do not disrupt the application to codes (and remain interesting in an analogue simulation context) but may reduce the required resources:

\begin{problem}\label{op: simple}
    To what extend can the construction be improved assuming that the target Hamiltonian is commuting and sparse? -- as is the case for most codes studied. 
\end{problem}

\begin{problem}\label{op: gs}
    To what extend can the construction be improved if we are solely interested in preserving the properties of the ground state?  
\end{problem}

To exemplify \cref{op: gs} consider the 4-body Hamiltonian $ H = Z_1Z_2Z_3 Z_4$, introduce three ancilla qubits $\hbt_a, \hbt_b, \hbt_c$, and consider the simulation Hamiltonian $H' = Z_1Z_2Z_a + Z_3Z_4Z_b + Z_aZ_bZ_c + Z_c$.
Then for any state in the ground space of $H'$, its restriction on the data qubits is also in the ground space of $H$.
Conversely, for any state in the ground space of $H$, there exists a state on the ancilla subsystem such that the state is in the groundspace of $H'$. 

This example falls outside of known universal simulation methods, and indeed it seems unlikely that simulating even just the groundstate of sparse commuting Hamiltonians could generally be done with constant interaction strength.
The argument for this pessimism takes inspiration from Appendix \ref{sec:no-go}: consider the target Hamiltonian $H = -ZZ -XX$, which stabilizes $\ket{00} + \ket{11}$.
If the two qubits are $\Omega(n)$ far from each other, the correlation between the two qubits has to decay slowly with the distance.
It follows from the exponential clustering theorem~\cite{Nachtergaele2010}, that the interaction strength $\mu$ of any simulation Hamiltonian \emph{with a unique ground state} reproducing the ground state of $H$ has to obey $\mu \in \tilde{\Omega}(n)$.
Although allowing degeneracy in the ancilla subsystem would circumvent this line of reasoning, the uniqueness is a fundamental resource in current constructions -- particularly in the eigenstate correspondence. 

We now turn to the question of how this construction could be used to reliably store information. In particular, it is clear that if the target Hamiltonian has an energy barrier, then the simulation Hamiltonian also does. To what extent are other properties preserved?

\begin{problem}
     the target Hamiltonian is quantum memory, does the simulation behave similarly? In general, does this construction preserve topological phases?
\end{problem}
\begin{problem}
    Can a decoder for the target Hamiltonian be mapped to a decoder for the simulation Hamiltonian? Can the simulation be used in a fault-tolerant context?
\end{problem}

Considering adiabatic computing, our results could be used to simulate arbitrary sparse Hamiltonians on these platforms even when their connectivity is limited.
Whether this can be utilised as a practical tool for interesting simulations on near-term hardware remains to be seen.
This would require keeping track of error terms explicitly to get a handle of the constants involved in our asymptotic scalings. 
\begin{problem}
    Are there any physically relevant sparse Hamiltonians where our protocol newly facilitates a practical simulation implementation on current quantum analogue platforms?
\end{problem}

Finally, it might be of independent interest to note that our long-range gadgets provide new proofs of known results, \cite{Oliveira2008}. Indeed since the $k$-local problem on sparse Hamiltonians is QMA-complete \cite{kitaev2002classical}, our result naturally shows that this hardness extends to 2D spin systems. 
\section{Hamiltonian simulation}\label{sect Ham sim}

Hamiltonian simulation, as referred to in this work, involves engineering a Hamiltonian that exhibits approximately the same physical properties as a target Hamiltonian, while the structure and features of the two Hamiltonians may be vastly different.
This section introduces some definitions and key results from Hamiltonian simulation literature that we will employ in this work. 

Formally an approximate simulation is defined as,
\begin{definition}[Approximate simulation \cite{Cubitt2019}]\label{defn approx sim}
We say that $H_\textup{sim}$ is a $(\Delta, \eta, \epsilon)$-simulation of $H_\textup{target}$ if there exists a local encoding $\mathcal{E}(M) = T \left(M\otimes P + \overline{M} \otimes Q \right)T^\dagger$ such that:
\begin{enumerate}[i.]
\item There exists an encoding $\tilde{\mathcal{E}}(M) = \tilde{T}\left(M \otimes P + \overline{M}\otimes Q \right)\tilde{T}^\dagger$ such that $\tilde{T}\idty\tilde{T}^\dagger$ is the projector into the low ($<\Delta$) energy subspace of $H_\textup{sim}$ and $\norm{\tilde{T}-T}_\infty\leq \eta$
\item $\norm{(H_\textup{sim})_{\leq \Delta}-\tilde{\mathcal{E}}(H_\textup{target})}_\infty\leq \epsilon$.
\end{enumerate}
An encoding is a map $\mathcal{E}:\textup{Herm}_n \mapsto\textup{Herm}_m$ of the form: $\tilde{\mathcal{E}}(M) = \tilde{T}\left(M \otimes P + \overline{M}\otimes Q \right)\tilde{T}^\dagger$ where $P$, $Q$ are orthogonal complementary projectors ($P + Q = \idty$), $T$ is an isometry and overline denote complex conjugation.
\end{definition}
This definition has three parameters describing the approximation in the simulation.
The physics of the target system is encoded in the subspace of $H_\text{sim}$ with energy less that $\Delta$, the high energy subspace simulator Hamiltonian contributes inaccuracies to the physics observed hence in good simulations $\Delta$ is taken large to minimise these effects.
$\eta$ describes the error in the eigenstates since the local encoding describing the simulation does not map perfectly into the low energy subspace but instead is close to a general encoding that does.
Finally $\epsilon$ describes the error in the eigenspectrum, a simple consequence of (ii) is that the low energy spectrum of the simulator is $\epsilon$-close to the spectrum of the target.
Given small $\epsilon$ and $\eta$ the eigenspectrum and corresponding states of the target are well approximated by the simulator.
In \cite{Cubitt2019} it is shown that approximate simulations preserves important physical quantities up to controllable errors:
\begin{lemma}[{\cite[Lem.~27, Prop.~28, Prop.~29]{Cubitt2019}}] \label{physical-properties}
Let $H$ act on $(\mathbb{C}^d)^{\otimes n}$.
  Let $H'$ act on $(\mathbb{C}^{d'})^{\otimes m}$, such that $H'$ is a $(\Delta, \eta, \epsilon)$-simulation of $H$ with corresponding local encoding $\mathcal{E}(M) = T(M \otimes P + \overline{M} \otimes Q)T^\dagger$.
  Let $p = \rank(P)$ and $q = \rank(Q)$.
  Then the following holds true.
  \begin{enumerate}[i.]
  \item Denoting with $\lambda_i(H)$ (resp.\ $\lambda_i(H')$) the $i$\textsuperscript{th}-smallest eigenvalue of $H$ (resp.\ $H'$), then for all $1 \leq i \leq d^n$, and all $(i-1)(p+q) \leq j \leq i (p+q)$, $|\lambda_i(H) - \lambda_j(H')| \leq \epsilon$.
  \item The relative error in the partition function evaluated at $\beta$ satisfies
\[
      \frac{|\mathcal{Z}_{H'}(\beta) - (p+q)\mathcal{Z}_H(\beta) |}{(p+q)\mathcal{Z}_H(\beta)} \leq \frac{(d')^m \ee^{-\beta \Delta}}{(p+q)d^n \ee^{-\beta \|H\|}} + (\ee^{\epsilon \beta} - 1).
\]
  \item For any density matrix $\rho'$ in the encoded subspace for which $\mathcal{E}(\idty)\rho' = \rho'$, we have
\[
      \|\ee^{-\ii H't}\rho'\ee^{\ii H't} - \ee^{-\ii \mathcal{E}(H)t}\rho'\ee^{\ii \mathcal{E}(H)t}\|_1 \leq 2\epsilon t + 4\eta.
\]
  \end{enumerate}
\end{lemma}

Perturbation gadgets offer a means to construct approximate simulations.
Perturbation theory, commonly employed in physics, provides an approximation technique for solving complex problems within specific regimes such as low energy and weak couplings. 
In practice, this is accomplished by introducing ancilla qubits with neighbouring interactions into the interaction graph of the Hamiltonian.
The gadget Hamiltonian consists of a `heavy' Hamiltonian than projects the ancilla qubits into their ground state in the low energy regime, and a `perturbative' Hamiltonian that couples the ancilla qubits to the original register and facilitate a new effective interaction.
Different gadgets are employed to reduce locality, remove crossings or reduce the degree of a qubit etc - see \cref{sect: old pert} for examples.
In general, they are second order simulations and require that $\epsilon,\eta = 1/\poly(\Delta)$.

\begin{lemma}[Second order simulation~\cite{Bravyi2014}]\label{second order sim}
The Hilbert space $\mathcal{H}$ is decomposed as $\mathcal{H} = \mathcal{H}_- \oplus \mathcal{H}_+$ with associated projectors $\Pi_-$ and $\Pi_+$.
The unperturbed Hamiltonian $H=\Delta H_0$ and the perturbation $V=H_1 + \sqrt{\Delta}H_2$ have support (in the $\Pi_+$, $\Pi_-$ basis),
\begin{align*}
&H_0 = \left(
\begin{array}{c|c}
(H_0)_{++} & 0 \\
\hline
 0 & 0
\end{array}
\right) \\
&H_1 = \left(
\begin{array}{c|c}
(H_1)_{++} & 0 \\
\hline
 0 & (H_1)_{--}
\end{array}
\right) \\
&H_2 = \left(
\begin{array}{c|c}
(H_2)_{++} & (H_2)_{+-} \\
\hline
(H_2)_{-+} & 0
\end{array}
\right) .
\end{align*}
Where $\lambda_\textup{min} ((H_0)_{++}) \geq 1$ and $\max \{ \norm{H_1},\norm{H_2}\}\leq \Lambda$.

Suppose $\exists$ an isometry $T$ s.t. $\Ima(T)= \mathcal{H}_-$ and 
\begin{equation*}\label{eqn 2nd order}
\norm{T H_\textup{target}T^\dagger - H_{1--} + H_{2-+}(H_{0++})^{-1}H_{2+-}}_\infty \leq \frac{\epsilon}{2},
\end{equation*}
then $\tilde{H}=H+V$ is a $(\Delta/2,\eta,\epsilon)-$simulation of $H_\textup{target}$ if $\Delta \geq O\left(\frac{\Lambda^6}{\epsilon^2}+\frac{\Lambda^2}{\eta^2} \right)$.
\end{lemma}

Approximate simulation is transitive as shown by the following lemma:
\begin{lemma}[{\cite[Lem.~25]{Cubitt2019}}] \label{lm transitive}
Let $A$, $B$, $C$ be Hamiltonians such that $A$ is a $(\Delta_A, \eta_A, \epsilon_A)$-simulation of $B$ and $B$ is a  $(\Delta_B, \eta_B, \epsilon_B)$-simulation of $C$.
Suppose $\epsilon_A,\epsilon_B\leq \norm{C}$ and $\Delta_B\geq \norm{C} + 2\epsilon_A + \epsilon_B$.
Then $A$ is a $(\Delta,\eta,\epsilon)$-simulation of $C$, where $\Delta\geq \Delta_B - \epsilon_A$,
\[\eta = \eta_A + \eta_B + O\left(\frac{\epsilon_A}{\Delta_B - \norm{C} + \epsilon_B} \right) \quad and \quad \epsilon = \epsilon_A +\epsilon_B + O \left(\frac{\epsilon_A \norm{C}}{\Delta_B - \norm{C} + \epsilon_B} \right).\]
\end{lemma}

Gadgets can also be applied in parallel (see \cite{Piddock2018}).
Given $n$ qubits partitioned into $(k+1)$ disjoint subsets labelled $S_{i}$ with $i = 1,2,...,k,k+1$. 
Denote by $H_0^{(i)}$ the heavy interaction only acting non-trivially on $S_i$ with $P_-^{(i)}$ the projector into its groundspace.
Let $V^{(i)}$ be the interaction term with only non-trivial action on $S_i\cup S_{k+1}$.
By \cite{Piddock2018} the effective Hamiltonian simulated by $V=\sum_i V^{(i)}$ and $H_0 = \sum_i H_0^{(i)}$ is the sum of the interactions simulated individually by $V^{(i)}$ and $H_0^{(i)}$.
Now $P_- = \prod_i P_-^{(i)}$ is the projector into the groundspace of $H_0$.
Morally applying gadgets to different interactions can be done in a single round of perturbation and therefore implementing multiple gadgets in parallel has a negligible effect on the strengthd required compared to multiple rounds.

\subsection{Previous simulation constructions}

Simulating general Hamiltonians using a Hamiltonian restricted to a 2D planar has been previously studied \cite{Oliveira2008,Cubitt2019}.
They show that simulating any Hamiltonian $H$ on $n$ qubits with a Hamiltonian on a square lattice of $\poly(n)$ qubits requires interaction strengths of $O\left(\left(n,\mu_0,\left(\frac{1}{\epsilon^2} + \frac{1}{\eta^2} \right) \right)^{\poly(n)}\right)$ where $\mu_0$ is the maximum interaction strength in $H$.
Interesting gadgets exist \cite{Cao2018} that circumvent exponential interaction strengths by using exponential ancillas and allowing the degree of the graph to increase exponentially. 
However simulating general sparse Hamiltonians using previous \emph{perturbative} techniques required exponential resources of some kind \cite{Oliveira2008,Cubitt2019,Cao2018,Piddock2020,Piddock2018}.

If we then consider a Hamiltonian from the restricted class that acts on a \emph{spatially sparse} interaction graph (\cref{defn spacial sparsity}), the interaction strengths required in the 2D lattice simulating Hamiltonian only need to be $O\left( \poly(n, \mu_o, \left(\frac{1}{\epsilon^2}+\frac{1}{\eta^2} \right))\right)$.

\begin{definition}[Spatial sparsity, \cite{Oliveira2008}]\label{defn spacial sparsity}
    A spatially sparse interaction graph $G$ on $n$ vertices is defined as a graph in which (i). every vertex participates in $O(1)$ edges, (ii). there is a straight-line drawing in the plane such that every edge overlaps with $O(1)$ other edges and the length of every edge is $O(1)$.
\end{definition}

History state constructions \cite{aharonov2018hamiltonians, Kohler2020} are \emph{non-perturbative} simulations satsifying first order simulation (see \cite{Bravyi2014} Lemma 4).
In \cite{Zhou} these techniques were used to construct the first geometrically local simulators of general local Hamiltonians with polynomial resources:
\begin{theorem}[Proposition 2 \cite{Zhou}]\label{thm Zhou}
    Given any $O(1)$-local $n$-qudit Hamiltonian with $\norm{H} = O(\mathrm{poly}(n))$, one can construct a spatially sparse 5-local Hamiltonian $H_\textup{circuit}$ that efficiently simulates $H$ to precision $(\Delta, \eta, \epsilon)$, with $\Delta =O(\epsilon^{-1}\norm{H}^2 + \eta^{-1}\norm{H]})$. 
    $H_\textup{circuit}$ has $O(\mathrm{poly}(n,\epsilon^{-1}))$ terms \textbf{and qubits} and interaction energy at most $O(\mathrm{poly}(n,\epsilon^{-1},\eta^{-1}))$.
\end{theorem}
This breakthrough exponentially improved the energy resource required for this protocol. 
Note also that this result simulates general local Hamiltonians whereas this paper is concerned with sparse Hamiltonians -- local Hamiltonians with the additional constraint that each qubit is qubit in the target Hamiltonian is acted on non-trivially by $O(1)$ terms.
However, the simulator Hilbert space in \cref{thm Zhou} depends on the accuracy of the simulation so increasing the precision of the simulator will require a new lattice and Hamiltonian, a practical limitation.
An advantage of perturbative simulations is that the error can be improved without changing the simulator lattice size or Hamiltonian structure and just requires increasing the strengths of the terms.

\section{Long-range gadget}\label{sect long range}

In this section we prove the main technical contribution of the work: a perturbation gadget to simulate a 2-local interaction with a 1D chain of $n$ 2-local interactions using polynomial interaction strengths -- see \cref{fg longrange}.

\begin{figure}[h!]
\centering
\includegraphics[trim={0cm 0cm 0cm 0cm},clip,scale=0.5]{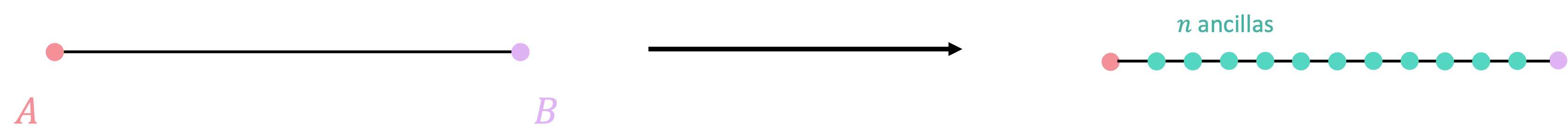}
\caption{Long range gadget}
\label{fg longrange}
\end{figure}

This is achieved by inserting $n$ ancillas in a single round of perturbation theory.
To facilitate an effective interaction through this chain of ancillas, the heavy Hamiltonian ground state must have long-range correlations in the ancilla chain while being itself geometrically local. 
Demonstrating such a Hamiltonian exists (additionally with reasonable interaction strengths and a spectral gap) is the key technical challenge to this gadget; taking inspiration from tensor network literature we will show such a parent Hamiltonian for the W state. 

\begin{definition}[Generalised W state]\label{defn W}
The generalised W state on $n$ qubits (denoted $\ket{W_n}$) is the normalised superposition of all computational basis states with a single $1$ and $(n-1)$ $0$'s,
\begin{equation*}
    \ket{W_n} : = \frac{1}{\sqrt{n}} \left(\ket{100...} + \ket{010...} + ... + \ket{0...001} \right).
\end{equation*}
\end{definition}

The first step in the proof is to construct a 2-local parent Hamiltonian for the generalised W state that will be the heavy Hamiltonian in our gadget. 

\subsection{Parent Hamiltonian of the W state}

We study a parent Hamiltonian of the generalised $W$ state, $H_W$, to use as a heavy Hamiltonian in our gadget construction. 
The construction of the gadget also requires that $H_W$ be gapped and have reasonably scaling interactions. 
The Hamiltonian we will need is encapsulated in the following theorem.

\begin{reptheorem}{thm parent Ham}[W state parent Hamiltonian]
Define the Hamiltonian acting on $n$ qubits 
\[ H_W = \Gamma \cdot \sum_{i = 1}^{n-1} P_{i,i+1}  +\idty -  \sum_{i = 1}^{ n}  (\idty - Z_i)/2 \] 
where $\Gamma \in O(\poly(n))$, $Z$ is the Pauli Z operator and $P$ is the two qubit projector,
\[  P_{i,i+1} = \begin{pmatrix}
0 & 0 & 0 &0\\
0 & \frac{1}{2} & -\frac{1}{2} &0 \\
0 & -\frac{1}{2} & \frac{1}{2} & 0\\
0 & 0 & 0 & 1\\
\end{pmatrix}\] 
acting on the $i$-th and $(i+1)$-th qubits (matrix in the computational basis).

Then $H_W$ satisfies the following properties:

    \begin{enumerate}[(i.)]
        \item $\ket{W_n}$ is the unique ground state of $H_W$, with eigenvalue $0$;
        \item $H_W$ is gapped.
    \end{enumerate}
\end{reptheorem}

For the proof of the above we first examine a simpler Hamiltonian, $H_{W,0}$.
The construction of this Hamiltonian was inspired by the literature on parent Hamiltonians of matrix product states (MPS) \cite{Sanz2016,Fernandez-Gonzalez}.
$H_{W,0}$ is a gapless uncle Hamiltonian of the W state with a degenerate ground state that will be a useful building block towards proving \cref{thm parent Ham}.

\begin{theorem}\label{thm HW0}
    Define a Hamiltonian acting on $n$ qubits
    \[ H_{\textup{W}, 0} = \sum_{i = 1}^{n-1} P_{i,i+1} \] 
where $P_{i,i+1}$ is the projector acting on qubits $i$ and $(i+1)$,
\[  P_{i,i+1} = \begin{pmatrix}
0 & 0 & 0 &0\\
0 & \frac{1}{2} & -\frac{1}{2} &0 \\
0 & -\frac{1}{2} & \frac{1}{2} & 0\\
0 & 0 & 0 & 1
\end{pmatrix}.\] 
This Hamiltonian, $H_{W,0}$
\begin{enumerate}[(i)]
    \item is frustration free with a ground state spanned by $\{\ket{0^{\otimes n}}, \ket{W_n}\}$;
    \item is gapless with a spectral gap $\Delta$ scaling as $\Delta = \Omega\left(\frac{1}{\poly(n)} \right)$.
\end{enumerate}
\end{theorem}
\begin{proof}
    Begin by writing the $\swap$ operation between qubits $i$ and $j$ as:
\begin{equation}
\swap_{i,j} :=  \begin{pmatrix}
1 & 0 & 0 &0\\
0 & 0 & 1 &0 \\
0 & 1 & 0 & 0\\
0 & 0 & 0 & 1\\
\end{pmatrix} \ .
\end{equation}
    With this notation,  rewrite  $H_{\textup{W}, 0} $ as follows:
    \begin{align}
        H_{\textup{W}, 0} &= \sum_{i = 1}^{n-1} P_{i,i+1} \\
        &=  \sum_{i = 1}^{n-1} (\idty - \swap_{i,i+1})/2 + \sum_{i = 1}^{n-1} \dyad{11}_{i,i+1} \\
        &= H_{\swap} + H_{\dyad{11}}.
    \end{align}
    First note that $[H_{\swap},H_{\dyad{11}}] \neq 0$ so in general the two components are not simultaneously diagonalisable. 

    Despite this, observe that $H_{W,0}$ cannot change the Hamming weight of a computational basis state and therefore $H_{W,0}$ is block diagonal in this basis.
    Additionally computational basis states with Hamming weight $\leq1$ are clearly in the kernel of $H_{\dyad{11}}$.
    Hence for states with Hamming weight $\leq 1$ the two Hamiltonians are simultaneously diagonalisable. 

    We are now in a position to prove the first statement.

    \paragraph{Proof of (i.)}
    First focus on the eigenspectrum of $H_{\swap}$.
Because $\swap$ has eigenvalues $\pm 1$, $(\idty - \swap_{i,j})/2$ is the projector on its $-1$ eigenspace.
Then by definition, the ground space of $H_{\swap}$ is spanned by states invariant under the permutation of two neighbouring qubits, i.e. they need to satisfy $\swap_{i,i+1} \ket{\phi} = \ket{\phi}$.

Since $\swap_{i,i+1} \swap_{i + 1,i+2} = \swap_{i,i+2}$, the ground space is invariant under arbitrary permutation of the qubits.
Conclude that any $\ket{\phi}$ in the ground space obeys,
\begin{equation}
\ket{\phi} \propto \sum_{\pi \in S_n} \pi \ket{\phi},
\end{equation}
where $S_n$ is a representation of the permutation group acting on $n$ qubits.
Hence, the frustration free groundspace of $H_\textup{SWAP}$ is $n$ dimensional and spanned by the uniform superposition of computational computational basis states with Hamming weight $w\in[1,n]$.

Since $H_\textup{SWAP}$ and $H_{\dyad{11}}$ are simultaneously diagonalisable on states with Hamming weight $w\leq 1$, conclude that $\ket{0^{\otimes n}}$ and $\ket{W_n}$ are frustration free groundstates of the total Hamiltonian.

We argue that these two states fully span the groundspace of $H$ by considering the subspaces spanned by computational basis states with Hamming weight $w>1$.
While the Hamiltonians do not commute on these spaces the energy of a state $\ket{\psi}$ is given by, 
\begin{equation}
    \bra{\psi}H\ket{\psi} = \bra{\psi}H_\textup{SWAP}\ket{\psi} + \bra{\psi}H_{\dyad{11}}\ket{\psi}.
\end{equation}
Let $\mathcal{F}$ be the set states that is not permutation invariant under $\text{SWAP}$ such that for all $\ket{x}\in \mathcal{F}:$ $\bra{x}H_\textup{SWAP}\ket{x}>0$.
Let $\mathcal{G}$ be states where $\exists$ $i\in[1,n]$ for all $\ket{y}\in \mathcal{G}$ such that $\dyad{11}_{i,i+1}\ket{y} = \ket{y}$ and therefore $\bra{y}H_{\dyad{11}}\ket{y}>0$.

If $\abs{\braket{y}{x}} >0$ where $\ket{y}$ is a state in the computational basis with Hamming weight at least 2, then there exists a state $\ket{y'} = \pi \ket{y}\in \mathcal{G}$, where $\abs{\braket{y'}{\phi}} > 0$.
Hence the intersection of $\mathcal{F}^c \cap \mathcal{G}^c = \emptyset$ when looking at the subspace with Hamming weight $w>1$ and all states not in $\mathrm{supp}\{\ket{0^{\otimes n}}, \ket{W_n} \}$ have finite energy with respect to $H_{W,0}$.

    \paragraph{Proof of (ii.)}
    
Throughout this section of the proof, $\mathcal{A}, \mathcal{B} \subset \Lambda$ will be two contiguous subset of qubits of the $1D$ line $\Lambda = [1,...,m]$ such that $\mathcal{A} \cup \mathcal{B} = \Lambda$.
Denote the Hamiltonian $H_{W,0}$ restricted to a subset of qubits $\mathcal{C}$ as $H_\mathcal{C}$.
We will show the gap is closing with polynomial scaling by employing the results from~\cite{Kastoryano2018}.
They introduce a measure of the overlap of ground states between different regions, 
\begin{equation}
    \delta(\mathcal{A},\mathcal{B}) : = \norm{(\Pi_\mathcal{A} - \Pi_{\mathcal{A}\cup \mathcal{B}})(\Pi_\mathcal{B}-\Pi_{\mathcal{A}\cup \mathcal{B}}}_\infty,
\end{equation}
where $\Pi_\mathcal{C}$ is the orthogonal projector onto the groundspace of the Hamiltonian $H_\mathcal{C}$.
Note that since $H$ in this case is frustration free so $\Pi_\mathcal{A}\Pi_{\mathcal{A}\cup \mathcal{B}} = \Pi_{\mathcal{A}\cup \mathcal{B}}\Pi_{\mathcal{A}} = \Pi_{\mathcal{A}\cup \mathcal{B}}$, and therefore, $\delta(\mathcal{A},\mathcal{B}) = \norm{\Pi_\mathcal{A}\Pi_\mathcal{B} - \Pi_{\mathcal{A}\cup \mathcal{B}}}$.
\cite{Kastoryano2018} equation 27 uses this measure to bound the first non-trivial eigenvalue of the Hamiltonian on $\mathcal{A}\cup \mathcal{B}$ by the lowest non-trivial eigenvalue of the Hamiltonian on $\mathcal{A}$ and $\mathcal{B}$ separately,
\begin{equation}\label{eqn 28 ref}
    \lambda_{\mathcal{A}\cup \mathcal{B}} \geq \frac{1-2\delta(\mathcal{A},\mathcal{B})}{2}\min(\lambda_\mathcal{A},\lambda_\mathcal{B}).
\end{equation}
Consequently demonstrating a polynomially decreasing gap reduces to calculating $\delta_{\mathcal{A},\mathcal{B}}$ for our Hamiltonian.

From (i.) the projector $\Pi_\mathcal{C}$ on the groundspace of the Hamiltonian, $H_{W,0}$ restricted to $\mathcal{C} \subset \Lambda$ is given by,
\begin{equation}\label{eqn projector of HW0}
    \Pi_\mathcal{C} = (\dyad{0}_\mathcal{C} + \dyad{W}_\mathcal{C}) \otimes \idty_{\Lambda \setminus 
    \mathcal{C}}.
\end{equation}

Write $\ket{W}_\mathcal{C} = \frac{1}{\sqrt{\abs{\mathcal{C}}}} \left(\ket{1000...}_\mathcal{C} + \ket{0100...}_\mathcal{C} + ... + \ket{0000...1}_\mathcal{C}\right) $ for the generalised W state on any subset $\mathcal{C} \subset \Lambda$.
Denote by $\abs{\mathcal{C}}$ the number of qubits in the subset $\mathcal{C}$.
The $W$ state on $\mathcal{A}$ can be expressed as
\begin{equation}
\ket{W}_\mathcal{A} = \frac{1}{\sqrt{\abs{\mathcal{A}}}}\left(\sqrt{\abs{\mathcal{A} \cap \mathcal{B}|}} \ket{W}_{\mathcal{A} \cap \mathcal{B}}\ket{0}_{\bar{\mathcal{B}}} + \sqrt{\abs{\bar{\mathcal{B}}}} \ket{0}_{\mathcal{A} \cap \mathcal{B}} \ket{W}_{\bar{\mathcal{B}}} \right),
\end{equation}
where $\bar{\mathcal{B}}$ denotes $\Lambda \setminus \mathcal{B}$.
Equivalently,
\begin{equation}\label{eqn 11}
\ket{W}_\mathcal{B} = \frac{1}{\sqrt{\abs{\mathcal{B}}}}\left(\sqrt{\abs{\mathcal{A} \cap \mathcal{B}}} \ket{W}_{\mathcal{A} \cap \mathcal{B}}\ket{0}_{\bar{\mathcal{A}}} + \sqrt{\abs{\bar{\mathcal{A}}}} \ket{0}_{\mathcal{A} \cap \mathcal{B}} \ket{W}_{\bar{\mathcal{A}}} \right),
\end{equation}
where note that $\bar{\mathcal{A}}\cap \mathcal{B} = \bar{\mathcal{A}}$ since $\mathcal{A}\cup \mathcal{B} = \Lambda$.

Let $a = \abs{\mathcal{A}}$, $b=\abs{\mathcal{B}}$, $\bar{a} = \abs{\Lambda \setminus \mathcal{A}} = \abs{\bar{\mathcal{A}}}$, $\bar{b} = \abs{\Lambda \setminus \mathcal{B}}= \abs{\bar{\mathcal{B}}}$ and $l = \abs{\mathcal{A} \cap \mathcal{B}}$.
There are two ways of expressing $\ket{W}_\Lambda$:
\begin{align}
    \ket{W}_\Lambda &= \frac{1}{\sqrt{m}}\left(\sqrt{a}\ket{W}_\mathcal{A}\ket{0}_{\bar{\mathcal{A}}} + \sqrt{\bar{a}}\ket{W}_{\bar{\mathcal{A}}}\ket{0}_\mathcal{A} \right)\\
    & = \frac{1}{\sqrt{m}}\left(\sqrt{b}\ket{W}_\mathcal{B}\ket{0}_{\bar{\mathcal{B}}}+\sqrt{\bar{b}}\ket{W}_{\bar{\mathcal{B}}}\ket{0}_\mathcal{B} \right).
\end{align}
Then
\begin{equation}
    \begin{multlined}
        \dyad{W}{W}_\Lambda = \frac{1}{m} \left( \sqrt{a\bar{b}} \ket{W}_\mathcal{A} \ket{0}_{\bar{\mathcal{A}}} \bra{W}_{\bar{\mathcal{B}}} \bra{0}_\mathcal{B} 
    + \sqrt{ab} \ket{W}_\mathcal{A}\ket{0}_{\bar{\mathcal{A}}} \bra{0}_{\bar{\mathcal{B}}} \bra{W}_\mathcal{B} \right.\\
    \left.+ \sqrt{\bar{a}\bar{b}}\ket{0}_\mathcal{A}\ket{W}_{\bar{\mathcal{A}}} \bra{W}_{\bar{\mathcal{B}}} \bra{0}_\mathcal{B} 
    + \sqrt{\bar{a}b} \ket{0}_\mathcal{A}\ket{W}_{\bar{\mathcal{A}}}\bra{0}_{\bar{\mathcal{B}}}\bra{W}_\mathcal{B}\right).
    \end{multlined}
\end{equation}

Using~\cref{eqn projector of HW0} we have:
\begin{equation}
    \Pi_\mathcal{A} \Pi_\mathcal{B} = \left[(\dyad{0}_\mathcal{A}+\dyad{W}_\mathcal{A})\otimes \idty_{\bar{\mathcal{A}}} \right] \cdot \left[(\dyad{0}_\mathcal{B}+\dyad{W}_\mathcal{B})\otimes \idty_{\bar{\mathcal{B}}} \right]
\end{equation}

Expanding the above and examining the terms individually,
\begin{equation}
\dyad{0}_\mathcal{A}\otimes \idty_{\bar{\mathcal{A}}} \cdot \idty_{\bar{\mathcal{B}}} \otimes \dyad{0}_\mathcal{B} = \dyad{0}_{\mathcal{A} \cup \mathcal{B}}
\end{equation}
\begin{equation}\label{eqn term 1}
    \left(\dyad{0}_\mathcal{A}\otimes \idty_{\bar{\mathcal{A}}}\right)  \cdot \left( \idty_{\bar{\mathcal{B}}} \otimes \dyad{W}_\mathcal{B} \right)=  (\ket{0}_\mathcal{A} \otimes \idty_{\bar{\mathcal{A}}} )(\bra{0}_A \otimes \idty_{\bar{\mathcal{A}}} ) (\idty_{\bar{\mathcal{B}}} \otimes \ket{W}_\mathcal{B})(\idty_{\bar{\mathcal{B}}} \otimes \bra{W}_\mathcal{B}).
\end{equation}
Using \cref{eqn 11} this can be rewritten
\begin{align}
\begin{multlined}
    (\bra{0}_\mathcal{A} \otimes \idty_{\bar{\mathcal{A}}} )  \  (\idty_{\bar{\mathcal{B}}} \otimes \ket{W}_\mathcal{B})
  \\
  =  (\bra{0}_\mathcal{A} \otimes \idty_{\bar{\mathcal{A}}} ) \ \left( \idty_{\bar{\mathcal{B}}} \otimes \frac{1}{\sqrt{b}}(\sqrt{l} \ket{W}_{A \cap \mathcal{B}}\ket{0}_{\bar{\mathcal{A}}} + \sqrt{\bar{a}} \ket{0}_{\mathcal{A} \cap \mathcal{B}} \ket{W}_{\bar{\mathcal{A}}} ) \right).
   \end{multlined}
\end{align}
Since $  (\bra{0}_\mathcal{A} \otimes \idty_{\bar{\mathcal{A}}} )\ \idty_{\bar{\mathcal{B}}} \otimes  \ket{W}_{\mathcal{A} \cap \mathcal{B}}\ket{0}_{\bar{\mathcal{A}}} = (\bra{0}_{\bar{\mathcal{B}}}\bra{0}_{\mathcal{A} \cap \mathcal{B}} \otimes \idty_{\bar{\mathcal{A}}} ) \  \idty_{\bar{\mathcal{B}}} \otimes \ket{W}_{\mathcal{A} \cap \mathcal{B}}\ket{0}_{\bar{\mathcal{A}}} = 0$, substituting into~\cref{eqn term 1} gives
\begin{align}
     \left(\dyad{0}_\mathcal{A}\otimes \idty_{\bar{\mathcal{A}}}\right)  \cdot\left( \idty_{\bar{\mathcal{B}}} \otimes \dyad{W}_\mathcal{B} \right)&=\sqrt{\frac{\bar{a}}{b}} \left(\ket{0}_\mathcal{A}\bra{0}_\mathcal{A}\otimes \idty_{\bar{\mathcal{A}}} \right)\left(\idty_{\bar{\mathcal{B}}} \otimes \ket{0}_{\mathcal{A}\cap \mathcal{B}}\ket{W}_{\bar{\mathcal{A}}} \right) \bra{W}_\mathcal{B}\\
     & =\sqrt{\frac{\bar{a}}{b}} (\ket{0}_\mathcal{A} \otimes\ket{W}_{\bar{\mathcal{A}}} )(\bra{0}_{\bar{\mathcal{B}}}\otimes \bra{W}_\mathcal{B}).
\end{align}

Similarly, we can obtain 
\begin{equation}
    \left(\dyad{W}_\mathcal{A}\otimes \idty_{\bar{\mathcal{A}}}\right)\cdot\left(\idty_{\bar{\mathcal{B}}} \otimes \dyad{0}_\mathcal{B}\right)  = \sqrt{\frac{\bar{b}}{a}} \ket{W}_\mathcal{A}\ket{0}_{\bar{\mathcal{A}}}\bra{W}_{\bar{\mathcal{B}}}\bra{0}_\mathcal{B}.
\end{equation}
Finally,
\begin{equation}
    \left(\dyad{W}_\mathcal{A}\otimes \idty_{\bar{\mathcal{A}}}\right)\cdot\left(\idty_{\bar{\mathcal{B}}} \otimes \dyad{W}_\mathcal{B}\right) = \sqrt{\frac{\bar{a}\bar{b}}{ab}} \ket{W}_\mathcal{A} \ket{W}_{\bar{\mathcal{A}}} \bra{W}_{\bar{B}}\bra{W}_{\mathcal{B}} + \frac{l}{\sqrt{ab}}\ket{W}_\mathcal{A}\ket{0}_{\bar{\mathcal{A}}}\bra{0}_{\bar{\mathcal{B}}}\bra{W}_\mathcal{B}.
\end{equation}

Compiling everything together, yields
\begin{equation}
\begin{multlined}
     \Pi_\mathcal{A} \Pi_\mathcal{B} = \dyad{0}_\Lambda 
     +  \sqrt{\frac{\bar{a}}{b}} \ket{0}_\mathcal{A} \ket{W}_{\bar{\mathcal{A}}} \bra{0}_{\bar{\mathcal{B}}} \bra{W}_\mathcal{B} 
     +  \sqrt{\frac{\bar{b}}{a}} \ket{W}_\mathcal{A}\ket{0}_{\bar{\mathcal{A}}}\bra{W}_{\bar{\mathcal{B}}}\bra{0}_\mathcal{B} \\    
     + \sqrt{\frac{\bar{a}\bar{b}}{ab}} \ket{W}_\mathcal{A} \ket{W}_{\bar{\mathcal{A}}} \bra{W}_{\bar{\mathcal{B}}} \bra{W}_{\mathcal{B}} + \frac{l}{\sqrt{ab}}\ket{W}_\mathcal{A}\ket{0}_{\bar{\mathcal{A}}}\bra{0}_{\bar{\mathcal{B}}}\bra{W}_\mathcal{B}.
    \end{multlined}
\end{equation}
Then recalling, 
\begin{align}
    \Pi_{\mathcal{A}\cup \mathcal{B}} &= \dyad{0}_{\mathcal{A}\cup \mathcal{B}}+\dyad{W}_{\mathcal{A}\cup \mathcal{B}}\\
    & =  \dyad{0}_{\mathcal{A}\cup \mathcal{B}} + \frac{1}{m} \left( \sqrt{a\bar{b}} \ket{W}_\mathcal{A} \ket{0}_{\bar{\mathcal{A}}} \bra{W}_{\bar{\mathcal{B}}} \bra{0}_\mathcal{B} 
    + \sqrt{ab} \ket{W}_A\ket{0}_{\bar{\mathcal{A}}} \bra{0}_{\bar{\mathcal{B}}} \bra{W}_\mathcal{B} \right.\notag\\
    & \hspace{1.5cm}\left.+ \sqrt{\bar{a}\bar{b}}\ket{0}_\mathcal{A}\ket{W}_{\bar{\mathcal{A}}} \bra{W}_{\bar{\mathcal{B}}} \bra{0}_\mathcal{B} 
    + \sqrt{\bar{a}b} \ket{0}_\mathcal{A}\ket{W}_{\bar{\mathcal{A}}}\bra{0}_{\bar{\mathcal{B}}}\bra{W}_\mathcal{B}\right)
\end{align}
finally, we can compute
\begin{equation}
    \begin{multlined}
        \Pi_{\mathcal{A} \cup \mathcal{B}} - \Pi_\mathcal{A} \Pi_\mathcal{B} =  \left(\frac{\sqrt{a \bar{b}}}{m} - \sqrt{\frac{\bar{b}}{a}} \right) \ket{W}_\mathcal{A}\ket{0}_{\bar{\mathcal{A}}}\bra{W}_{\bar{\mathcal{B}}}\bra{0}_\mathcal{B} 
    + \left( \frac{\sqrt{ab}}{m} - \frac{l}{\sqrt{ab}}\right) \ket{W}_\mathcal{A}\ket{0}_{\bar{\mathcal{A}}}\bra{0}_{\bar{\mathcal{B}}}\bra{W}_\mathcal{B} \\
    + \left( \frac{\sqrt{\bar{a}b}}{m} - \sqrt{\frac{\bar{a}}{b}} \right)  \ket{0}_\mathcal{A} \ket{W}_{\bar{\mathcal{A}}} \bra{0}_{\bar{\mathcal{B}}} \bra{W}_\mathcal{B}
    +  \frac{\sqrt{\bar{a}\bar{b}}}{m} \ket{0}_\mathcal{A}\ket{W}_{\bar{\mathcal{A}}} \bra{W}_{\bar{\mathcal{B}}} \bra{0}_\mathcal{B} 
    -  \frac{\sqrt{\bar{a}\bar{b}}}{\sqrt{ab}} \ket{W}_\mathcal{A} \ket{W}_{\bar{\mathcal{A}}} \bra{W}_{\bar{\mathcal{B}}} \bra{W}_{\mathcal{B}}.
    \end{multlined}
\end{equation}

Now pick $\mathcal{A}$ and $\mathcal{B}$ such that $l = \gamma m, a = b = \frac{1 + \gamma}{2}m, \bar{a} = \bar{b} = \frac{1-\gamma}{2}m$, for $\gamma \in (0,1)$.
By the triangle inequality,
\begin{align}
    \norm{\Pi_{\mathcal{A} \cup \mathcal{B}} - \Pi_\mathcal{A} \Pi_\mathcal{B} }\leq &\abs{\frac{\sqrt{a \bar{b}}}{m} - \sqrt{\frac{\bar{b}}{a}}} + \abs{\frac{\sqrt{ab}}{m} - \frac{l}{\sqrt{ab}}} + \abs{ \frac{\sqrt{\bar{a}b}}{m} - \sqrt{\frac{\bar{a}}{b}} } + \abs{\frac{\sqrt{\bar{a}\bar{b}}}{m}} + \abs{ \frac{\sqrt{\bar{a}\bar{b}}}{\sqrt{ab}}}\\
     \leq&\abs{\frac{\sqrt{(1+\gamma)(1-\gamma)}}{2}-\sqrt{\frac{1-\gamma}{1+\gamma}}} + \abs{\frac{1+\gamma}{2} - \frac{2\gamma}{1+\gamma}} + \frac{1-\gamma}{2} \notag\\
    &+ \abs{\frac{\sqrt{(1-\gamma)(1+\gamma)}}{2} - \sqrt{\frac{1-\gamma}{1+\gamma}} }+ \frac{1-\gamma}{1+\gamma}\\
    \leq & 5 \frac{1-\gamma}{1+\gamma},
\end{align}
where the final line follows from, $\frac{1-\gamma}{2}\leq \frac{1-\gamma}{1+\gamma}$, $\frac{1+\gamma}{2}-\frac{2\gamma}{1+\gamma}$ and $\abs{\frac{\sqrt{(1+\gamma)(1-\gamma)}}{2}-\sqrt{\frac{1-\gamma}{1+\gamma}}}\leq \frac{1-\gamma}{1+\gamma}$.

Then consider~\cref{eqn 28 ref} and note that if $\gamma = \gamma^*>9/11$ then $\delta(\mathcal{A},\mathcal{B})<1/2$ and,
\begin{equation}
  \lambda_{\mathcal{A}\cup \mathcal{B}} \geq \epsilon \min (\lambda_\mathcal{A}, \lambda_\mathcal{B})  ,
\end{equation}
where $\epsilon = \frac{1}{2} - 5 \frac{1-\gamma^*}{1+\gamma^*}$ is a constant $ > 0 $.

This expression can be used to obtain a lower bound on $\lambda_\Lambda$.
Consider $f(n) = \epsilon f(\frac{1+\gamma^*}{2}n)$, with $f(O(1)) \in \Omega(1)$, then solving the recurrence relation gives 
\begin{equation}
f(n) \in \Omega\left((\epsilon^{-1})^{\log(n)/\log((1 + \gamma)/2)}\right) 
\end{equation}
which gives $\lambda_\Lambda \geq f(n) \in \Omega(1/\poly(n))$ and the result.
\end{proof}

Note that by computing the operator norm exactly -- see Appendix \ref{sec:exact-norm} -- the exponent is at least $-6.13$ i.e. $\lambda_\Lambda \in \Omega(n^{-6.13})$.

$H_{\textup{W}, 0}$ is a step towards a good candidate for $H_0$ in \cref{second order sim}: $\ket{W_n}$ state is in the ground space and while it is gapless, we have some control of how quickly the gap closes with $n$.
However, the degeneracy in the ground state between $\ket{0^{\otimes n}}$ and $\ket{W_n}$ makes it unsuitable to be used in \cref{second order sim} with a simple isometry.
Constructing the W state parent Hamiltonian requires a final step to lift this degeneracy.

\begin{theorem}[W state Parent Hamiltonian]\label{thm parent Ham}
Define the Hamiltonian acting on $n$ qubits 
\[ H_W = \Gamma \cdot \sum_{i = 1}^{n-1} P_{i,i+1}  +\idty -  \sum_{i = 1}^{ n}  (\idty - Z_i)/2 \] 
where $\Gamma \in O(\poly(n))$, $Z$ is the Pauli Z operator and $P$ is the two qubit projector,
\[  P_{i,i+1} = \begin{pmatrix}
0 & 0 & 0 &0\\
0 & \frac{1}{2} & -\frac{1}{2} &0 \\
0 & -\frac{1}{2} & \frac{1}{2} & 0\\
0 & 0 & 0 & 1\\
\end{pmatrix}\] 
acting on the $i$-th and $i+1$-th qubits (matrix in the computational basis).

Then $H_W$ satisfies the following properties:

    \begin{enumerate}[(i.)]
        \item $\ket{W_n}$ is the unique ground state of $H_W$, with eigenvalue $0$;
        \item $H_W$ is gapped.
    \end{enumerate}
\end{theorem}

\begin{proof}
Recall from \cref{thm HW0} that $H_W$ can be rewritten as
\begin{equation}
H_W = \Gamma \cdot H_{W,0} + H_Z,
\end{equation}
where $H_Z = \idty -  \sum_{i = 1}^{i = n}  (\idty - Z_i)/2$.
Since note that $[H_{W,0},H_Z] = 0$ their spectra can be analysed separately.

The behaviour of $H_{W,0}$ is sufficiently well understood by \cref{thm HW0}, so we are left to address the rightmost part of the expression.
$H_{Z}$ acting on the groundstates of $H_{W,0}$ is,
\begin{align}
&\left (\idty -  \sum_{i = 1}^{ n}  (\idty_i - Z_i)/2  \right )\ \ket{0^{\otimes n}} = \ket{0^{\otimes n}}\\
&\left (\idty -  \sum_{i = 1}^{n}  (\idty_i - Z_i)/2  \right ) \ \ket{W_n} = 0.
\end{align}

This yeilds 
\begin{align}
&H_W  \ket{0^{\otimes n}} = \ket{0^{\otimes n}} \\
&H_W  \ket{W_n} = 0.
\end{align}
However, we are no longer dealing with a frustration free Hamiltonian and $H_Z$ clearly has negative eigenvalues.
So while we have lifted the degeneracy of $\{\ket{W_n},\ket{0}\}$ we need to examine whether $\ket{W_n}$ with energy $0$ is the groundstate or there is an eigenvalue smaller than $0$.

Since $\norm{\sum_{i = 1}^{n}  (\idty_i - Z_i)/2 }= n$, for any state $\rho$ not in the support of $\textup{span}\{\ket{0^{\otimes n}}, \ket{W_n}\}$, we have:
\begin{align}
    \tr(H_W \rho) &= \Gamma \cdot \tr(H_{\textup{W}, 0} \rho) - \tr(\left( - \idty + \sum_i (\idty_i - Z_i)/2 \right ) \rho) \\
    &\geq \Gamma \Delta + 1 - n \label{eqn lower bound 3 eigen}
\end{align}
where $\Delta$ is the spectral gap of $H_{W,0}$.

So long as $\Gamma\Delta > n$ the Hamiltonian has a constant spectral gap and $\ket{W_n}$ is the unique groundstate. 
We know from~\cref{thm HW0} that $\Delta \in \Omega(1/\poly(n))$.
Therefore one can chose a $\Gamma \in O(\poly(n))$ such that $\Gamma\Delta > n$ and the two statements are verified. 
\end{proof}

\subsection{Gadget construction}

Before proving the theorem we need a technical lemma bounding a coefficient appearing in the gadget Hamiltonian.
\begin{lemma}\label{lm bounding C}
    Define,
\[C := \frac{1}{n}\left(2 + \sum_{a,b=1}^n \bra{a,i}(H_{W++})^{-1}\ket{b,j} + \sum_{a',b'=1}^n \bra{a',j}(H_{W++})^{-1}\ket{b',i} \right)\]
    where $\ket{a, i}$ denotes the computational basis state on $n$ qubits with 1 at positions $\{a,i\}$ and 0 elsewhere.
    $H_W$ is the Hamiltonian from \cref{thm parent Ham} acting on $n$ qubits with spectral gap $\Delta$,
\[ H_W = \Gamma \cdot \sum_{i = 1}^{n-1} P_{i,i+1}  +\idty -  \sum_{i = 1}^{ n}  (\idty - Z_i)/2 \]
    and $H_{W++}= \Pi_+H_{W}\Pi_+$ with $\Pi_+ = \idty - \dyad{W_n}$.
    Then $\frac{1}{n}\leq C,$ if $(\Gamma \Delta +1 )>5n$.
\end{lemma}
\begin{proof}
Denote the unnormalised state $\ket{\psi} = \sum_{a=1}^n \ket{a,i}$ and note that $\swap_{i,j}\ket{\psi} = \sum_{b=1}^n \ket{b,j}$.
Let
\begin{align}
    P_- & = \frac{1}{2}(\idty_{i,j} - \swap_{i,j})\otimes \idty_{\{n\}/i,j}\\
    P_+ & = \frac{1}{2}(\idty_{i,j} + \swap_{i,j})\otimes \idty_{\{n\}/i,j},
\end{align}
be projectors onto the $-1$, $+1$ eigenspace of $\swap_{i,j}$ respectively.

We can then write,
\begin{align}
    \sum_{a,b=1}^n& \bra{a,i}(H_{W++})^{-1}\ket{b,j} + \sum_{a',b'=1}^n \bra{a',j}(H_{W++})^{-1}\ket{b',i} \notag\\
    & = \bra{\psi}(H_{W++})^{-1}\swap_{i,j}\ket{\psi} +\bra{\psi}\swap_{i,j}(H_{W++})^{-1}\ket{\psi}\\
    & = \bra{\psi}(H_{W++})^{-1}P_+ + P_+ (H_{W++})^{-1} \ket{\psi} - \bra{\psi}(H_{W++})^{-1}P_- + P_- (H_{W++})^{-1} \ket{\psi}\\
    & = \bra{\psi}(H_{W+'+'})^{-1}P_+ + P_+ (H_{W+'+'})^{-1} \ket{\psi} - \bra{\psi}(H_{W+'+'})^{-1}P_- + P_- (H_{W+'+'})^{-1} \ket{\psi}.
\end{align}
In the final line, $\Pi_+$ is replaced with $\Pi_{+'}$ the projector onto the orthogonal subspace of $\{\ket{0},\ket{W_n}\}$ which does not affect the values since $\bra{\psi}\ket{W_n} = \bra{\psi}\ket{0} = \bra{\psi}P_\pm\ket{W_n} = \bra{\psi}P_\pm\ket{0} = 0$.

Note that $A_{\pm} =(H_{W+'+'})^{-1}P_\pm + P_\pm (H_{W+'+'})^{-1} $ has largest eigenvalue (in absolute value) upper bounded by,
\begin{align}
    \norm{A_{\pm}}_\infty &\leq \norm{(H_{W+'+'})^{-1}P_\pm}_\infty + \norm{P_\pm (H_{W+'+'})^{-1}}_\infty\\
    & \leq 2 \norm{P_\pm}_\infty \norm{(H_{W+'+'})^{-1}}_\infty.
\end{align}
Then use work from~\cref{thm parent Ham} to upper bound $\norm{(H_{W+'+'})^{-1}}$.
Write $H_W = \sum_{i=0}^{2^n} \lambda_i \dyad{i}$ in its spectral decomposition where $\lambda_0 = 0$ is the ground state eigenvalue and $\lambda_1$ is the eigenvalue associated with $\ket{0^{\otimes n}}$.
Then the restricted inverse can be written as, $(H_{W+'+'})^{-1} = \sum_{i=2}^{2^n}\frac{1}{\lambda_i}\dyad{i}$.
The largest eigenvalue of $(H_{W+'+'})^{-1}$ is $\frac{1}{\lambda_2}$ which from~\cref{eqn lower bound 3 eigen} is lower bounded by,
\begin{equation}
    \lambda_2 \geq \Gamma\Delta +1 - n.  
\end{equation}
From the theorem statement $\Gamma > \frac{(5n-1)}{\Delta}$ -- note this is possible while $\Gamma\in O(\poly(n))$ since $\Delta \in \Omega (1/\poly(n))$ (see~\cref{thm parent Ham}) -- therefore $\norm{(H_{W+'+'})^{-1}}_\infty = \frac{1}{\lambda_2} \leq \frac{1}{4n}$.
Substituting into the above yeilds, 
\begin{equation}
    \norm{A_{\pm}} \leq 2 \cdot 1 \cdot \frac{1}{4n} = \frac{1}{2n}.
\end{equation}

Therefore 
\begin{align}
    \sum_{a,b=1}^n& \bra{a,i}(H_{W++})^{-1}\ket{b,j} + \sum_{a',b'=1}^n \bra{a',j}(H_{W++})^{-1}\ket{b',i} \notag\\
    & = \bra{\psi}A_+\ket{\psi} - \bra{\psi}A_-\ket{\psi}\\
    &\geq -2\braket{\psi}\norm{A_\pm}\\
    & \geq -1,
    \end{align}
in the final line since $\ket{\psi}$ is un-normalised $\braket{\psi} = n$.
Finally,
\begin{align}
    C \geq \frac{1}{n}(2-1) = \frac{1}{n}.
\end{align}
\end{proof}

\begin{theorem}[Long-range gadget]\label{long-range}
The Hamiltonian,
\[H_\textup{target} = H_\textup{else} + P_A \otimes P_B,\]
is $(\Delta/2,\eta,\epsilon)$-simulated by the $n$ long-range Hamiltonian , $\tilde{H}=H+V$, (qubits $A$ and $B$ are connected via a 1d chain of $n$ qubits) where
\begin{align*}
&H = \Delta H_W\\
&V = H_1 + \sqrt{\Delta}H_2\\
&H_1 = H_\textup{else}+ \frac{1D}{2C}(P_A^2 \otimes P_B^2) \\
&H_2 =\frac{1}{\sqrt{C}}\left( P_A \otimes X_i - P_B \otimes X_j \right),
\end{align*}

for any choice of $i,j\in[1,n]$ with $i\neq j$ in particular when $i=1$ and $j = n$ so that the simulating Hamiltonian is $n$ long range.
Where,
\[C := \frac{1}{n}\left(2 + \sum_{a,b=1}^n \bra{a,i}(H_{W++})^{-1}\ket{b,j} + \sum_{a',b'=1}^n \bra{a',j}(H_{W++})^{-1}\ket{b',i} \right)\]
\[
D: = 2 \bra{W_n}X_i (H_{W++})^{-1} X_i \ket{W_n}
\]
and can be chosen to be $C\geq \frac{1}{n}$, $D\leq2$ while $\norm{H_W}\in O(
\poly(n))$.
\end{theorem}

\begin{proof}
Let $\Pi_- = \ket{W_n}\bra{W_n}$ and $\Pi_+ = \idty - \Pi_-$ be its complement.
\cref{thm parent Ham} shows that $H_W$ is block diagonal in the $\Pi_+,\Pi_-$ basis with $(H_W)_{--} = 0$ (since $\ket{W_n}$ is the unique ground state with eigenvalue 0) and $\lambda_\textup{min}((H_W)_{++})\geq 1$ (since there is constant spectral gap).
Hence the conditions of \cref{second order sim} are satisfied and the task reduces to showing that,
\begin{equation}
\norm{T H_\textup{target}T^\dagger - H_{1--} + H_{2-+}(H_{W++})^{-1}H_{2+-}}_\infty \leq \frac{\epsilon}{2}.
\end{equation}

We start by examining the final term in the normed expression,
\begin{align}\label{eqn key term}
H_{2-+}(H_{W++})^{-1}H_{2+-} =& \frac{1}{C}\left[ \mathbf{1}_{AB}\otimes \left(\Pi_-X_i\Pi_+ H_W^{-1} \Pi_+X_i \Pi_- + \Pi_-X_j\Pi_+H_W^{-1}\Pi_+X_j\Pi_-\right)\right.\notag \\
& \left.- P_A\otimes P_B\otimes \left(\Pi_-X_i\Pi_+H_W^{-1}\Pi_+X_j\Pi_- + \Pi_-X_j\Pi_+H_W^{-1}\Pi_+X_i\Pi_- \right) \right].
\end{align}

A single Pauli $X$ acting on $\ket{W_n}$,
\begin{equation}
    X_i \ket{W_n} = \frac{1}{\sqrt{n}} \left( \sum_{a=1}^n\ket{a,i} +\ket{0^{\otimes n}}  \right)
\end{equation}
is completely contained in $\mathcal{H}_+$, i.e. $\Pi_+ X_i \ket{W_n} = X_i \ket{W_n}$.

The Hamiltonian written in its spectral decomposition is $H_W = \sum_{l} \lambda_l \ket{l}\bra{l}$ where $\{\ket{l}\}$ is an orthogonal basis of $\mathcal{H}_+$.
Then,
\begin{align}
    &\Pi_- X_i \Pi_+ H_0^{-1} \Pi_+ X_j \Pi_- \notag\\
    & = \ket{W_n}\bra{W_n} X_i \Pi_+ \sum_l \lambda_l^{-1} \ket{l}\bra{l} \Pi_+ X_j \ket{W_n}\bra{W_n}\\
    & = \frac{1}{n}\ket{W_n}\left( \sum_{a=1}^n\bra{a,i} +\bra{0^{\otimes n}}\right) \sum_l \lambda_l^{-1} \ket{l}\bra{l} \left( \sum_{b=1}^n\ket{b,i} +\ket{0^{\otimes n}}\right)\bra{W_n}\\
    & = \frac{1}{n}\ket{W_n}\left( \sum_{a=1}^n\bra{a,i} +\bra{0^{\otimes n}}\right) \left(\lambda_1^{-1} \ket{0^{\otimes n}}\bra{0^{\otimes n}} + \sum_{l=2} \lambda_l^{-1} \ket{l}\bra{l}\right) \left( \sum_{b=1}^n\ket{b,i} +\ket{0^{\otimes n}}\right)\bra{W_n}\\
    & = \frac{1}{n}\left(\frac{1}{\lambda_1}\abs{\braket{0^{\otimes n}}}^2 + \sum_{a,b=1}^n \bra{a,i}(H_{W++})^{-1}\ket{b,j} \right)\ket{W_n}\bra{W_n}\\
    & = \frac{1}{n}\left(1 + \sum_{a,b=1}^n \bra{a,i}(H_{W++})^{-1}\ket{b,j} \right)\Pi_-.
\end{align}
Similarly, $\Pi_- X_j \Pi_+ H_0^{-1} \Pi_+ X_i \Pi_- =\frac{1}{n}\left(1 + \sum_{a,b=1}^n \bra{a,j}(H_{W++})^{-1}\ket{b,i} \right)\Pi_- $.

Substituting into~\cref{eqn key term}
\begin{equation}
    \left(\Pi_-X_i\Pi_+H_W^{-1}\Pi_+X_j\Pi_- + \Pi_-X_j\Pi_+H_W^{-1}\Pi_+X_i\Pi_- \right) = C \Pi_-,
\end{equation}
where~\cref{lm bounding C} gives the lower bound on $C$ quoted in the theorem.

Setting $i = j$ gives
\begin{align}
    \left(\Pi_-X_i\Pi_+ H_W^{-1} \Pi_+X_i \Pi_- + \Pi_-X_j\Pi_+H_W^{-1}\Pi_+X_j\Pi_-\right) &= 2 \bra{W_n}X_i (H_{W++})^{-1} X_i \ket{W_n} \Pi_-\\
    & =D \Pi_-,
\end{align}
where $D\leq 2$ since we have chosen $\Gamma \Delta - n >1$ so that the spectral gap of $H_W$ is $\geq 1$.

Defining the isometry $T: \ket{\psi}_{AB} \mapsto \ket{\psi}_{AB}\otimes \ket{W_n} $, and substituting the above into \cref{eqn 2nd order}
\begin{equation}
\begin{multlined}
\norm{T H_\textup{target}T^\dagger - H_{1--} + H_{2-+}(H_{W++})^{-1}H_{2+-} }\\
 =(H_\textup{else} + P_A\otimes P_B) \otimes \Pi_- - H_\textup{else} \otimes \Pi_- - \frac{D}{C}\mathbf{1}_{AB}\otimes \Pi_- + \frac{D}{C} \mathbf{1}_{AB}\otimes \Pi_- \\
 - P_A\otimes P_B \otimes \Pi_- = 0.
\end{multlined}
\end{equation}
\cref{eqn 2nd order} is therefore satisfied for all $\epsilon \geq 0$.
So, provided a $\Delta$ is picked which satisfies the conditions of \cref{second order sim}, $\tilde{H}$ is a $(\Delta/2,\eta,\epsilon)$-simulation of $H_\textup{target}$.
\end{proof}

Note that since $C \geq \frac{1}{n}$ and $D\leq 2$ the interaction strengths appearing in the perturbation are upper bounded,
\begin{align}
    \norm{H_1} &\leq \norm{H_\textup{else}} + 2n\\
    \norm{H_2} &\leq 2\sqrt{n}.
\end{align}

\section{Improved Hamiltonian Simulation Protocol}

The new method for simulating $n$-long range 2-qubit interactions using the gadget in \cref{long-range} can now be employed in a sequence of gadgets to localise a general Hamiltonian. 
Previously, simulating a $n$-long range 2-qubit interaction with a chain of $n$ 2-qubit interactions required recursive applications of the subdivision gadget.
Each time the subdivision gadget is applied the length of the interaction is halved and one ancillary qubit is inserted.
It therefore requires $O(\log n)$ rounds of perturbation theory to reduce an edge of length $O(n)$, to $O(n)$ edges of length $O(1)$.
$O(\log(n))$ rounds of perturbation theory in turn requires simulation interaction strengths scaling as $O(\exp(n))$.
With the new gadget introduced, we can perform the same subroutine with only $O(\poly(n))$ strengths.
\cref{table sim} gives a summary of the steps involved in constructing this simulation.
The rest of the localisation recycles the results and techniques of \cite{Oliveira2008, Cubitt2019}.

\begin{table}[h!]
\small
\begin{center}
\begin{tabular}{l|l|l|l|l|}
\cline{2-5}
& \cellcolor[HTML]{EFEFEF}  Action & \cellcolor[HTML]{EFEFEF} Gadgets & \cellcolor[HTML]{EFEFEF}\# rounds & \cellcolor[HTML]{EFEFEF} \# ancillas \\ \hline
\multicolumn{1}{|l|}{\textbf{Step 1} }                   &     $\kappa$-local $\mapsto$                     &         Subdivision \& 3-to-2          &           $O(\log(\kappa))$               &        $O(n\kappa\delta)$                  \\ 
\multicolumn{1}{|l|}{\multirow{-2}{*}{}} & \multirow{-2}{*}{}  \hspace{-2.5mm} 2-local    & \multirow{-2}{*}{}       & \multirow{-2}{*}{}       & \multirow{-2}{*}{}       \\ \hline
\multicolumn{1}{|l|}{\textbf{Step 2} }                   &     Degree-$O(\delta)$ $\mapsto$                     &         Triangle          &           $O(\log(\delta))$               &        $O(n\kappa\delta^2)$                  \\ 
\multicolumn{1}{|l|}{\multirow{-2}{*}{}} & \multirow{-2}{*}{}  \hspace{-2.5mm} Degree-4    & \multirow{-2}{*}{}       & \multirow{-2}{*}{}       & \multirow{-2}{*}{}       \\ \hline
\multicolumn{1}{|l|}{\textbf{Step 3a} }                   &          Interaction length $O(n\kappa \delta)$ $\mapsto$                &      \textbf{Mthd 1}: Subdivision                   &         $O(\log(n\kappa\delta))$                 &         $O(n\kappa\delta)$                 \\ \cline{3-5} 
\multicolumn{1}{|l|}{\multirow{-2}{*}{}} & \multirow{-2}{*}{ } \hspace{-2.5mm}Interaction length $O(1)$  &   \textbf{Mthd 2}: Long-range                  &      $ O(1)   $                &            $O(n\kappa\delta)$               \\ \hline
\multicolumn{1}{|l|}{\textbf{Step 3b} }                   &     $O(n^2\kappa^2\delta^2)$ crossings $\mapsto$                     &         Crossing          &           $O(1)$               &        $O(n^2\kappa^2\delta^2)$                  \\ 
\multicolumn{1}{|l|}{\multirow{-2}{*}{}} & \multirow{-2}{*}{}  \hspace{-2.5mm} No crossings    & \multirow{-2}{*}{}       & \multirow{-2}{*}{}       & \multirow{-2}{*}{} \normalsize \\ \hline
\end{tabular}
\caption{Simulation sequence to transform a graph on $n$ qubits that is $\kappa$-local and has maximum degree $\delta$ into a geometrically 2-local graph on a $O(n\kappa\delta)\times O(n\kappa\delta)$ square lattice. To highlight the contribution of this work, in Step 3a we include two methods for the reduction : the previous standard using the subdivision gadget recursively; and our new gadget.}\label{table sim}
\end{center}
\end{table}

\begin{theorem}[Simulating general Hamiltonians]\label{thm general Ham}
    Given a Hamiltonian, $H = \sum_i h_i$ where $h_i$ is a Pauli rank one operator.
    $H$ acts on $n$ qubits and is $\kappa$-local (i.e. $\forall i$, $h_i$ only has non-trivial support on $\kappa$ qubits) and the maximum degree is $\delta$ (i.e. each qubit is involved in $\leq\delta$ non-trivial interactions).

$\exists$ a nearest neighbour Hamiltonian acting on a 2D lattice of $N = O(n^2\kappa^2\delta^2)$ qubits, $\tilde{H} = \sum_j \tilde{h}_j$,
that is a $(\Delta, \eta, \epsilon)$-simulation of $H$ with,
\begin{align*}
    \mu &\in O\left(\left[\poly(n,\delta,\kappa) \mu_0 \left(\frac{1}{\epsilon^2} + \frac{1}{\eta^2} \right)  \right]^{6(\log(\kappa) + \log(\delta) + O(1))} \right) \\
    \Delta &= \mu/2
\end{align*}
where $\mu = \max_j\norm{\tilde{h}_j}$, and $\mu_0 = \max_i \norm{h_i}$.
Further, the encoding isometry has the form $T : \ket{\psi} \longrightarrow \ket{\psi} \otimes \ket{\textup{anc}}$ for some state of the ancillary qubits $\ket{\textup{anc}}$. 
\end{theorem}

\begin{proof}
This proof goes via applying a sequence of perturbation gadgets to the target Hamiltonian, $H= H_\textup{target}$.
We remind the reader of \cref{second order sim} as we will repeatedly use it: $\tilde{H} = H + V$ will $(\Delta/2,\epsilon,\eta)$ - simulate $H_\textup{target}$ if, 
\begin{equation}
    \Delta \geq O \left(\frac{\Lambda^6}{\epsilon^2} + \frac{\Lambda^2}{\eta^2} \right),
\end{equation}
where $\Lambda \geq \max \{\norm{H_1}, \norm{H_2} \}$.

The sequence of perturbation gadgets we apply to $H_\textup{target}$ is as follows, where we employ that approximate simulation is transitive (\cref{lm transitive}) and disjoint gadgets can be applied in parallel.

\paragraph{Step 1 - Reduce locality}
$H_\textup{target}$ contains $O(n\delta)$ operators of weight $O(\kappa)$.

We will use the subdivision gadget recursively followed by a single application of the 3-to-2 gadget if necessary (see \cref{sect: old pert}) to obtain a two-local simulator Hamiltonian $H'$. 
See~\cref{fg step1} for reference. 

Reducing the locality from $\kappa$ to 2 local requires $O(\log \kappa)$ rounds of perturbation and $O(\kappa)$ ancillas. 
Applying this step to every interaction in $H_\textup{target}$ requires a total of $O(n\delta\kappa)$ ancillas.

\begin{figure}[h!]
\centering
\includegraphics[trim={0cm 0cm 0cm 0cm},clip,scale=0.4]{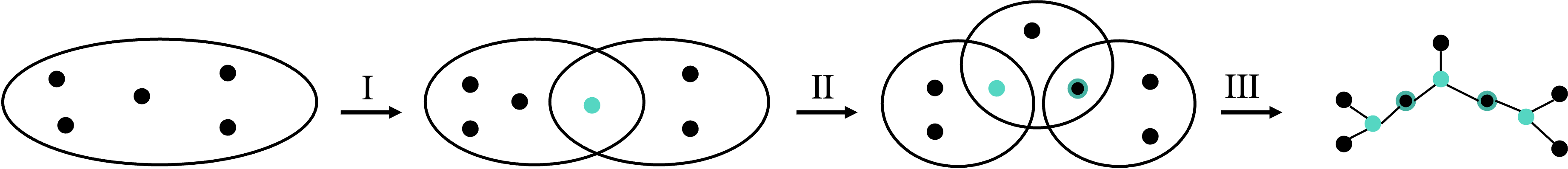}
\caption{Example of reducing locality. Given a $\kappa = 5$-local interaction on the Hamiltonian hypergraph. Applying the subdivision gadget from~\cref{subdivision} reduces the locality to $\ceil{5/2} = 3$ (labelled I). Applying the subdivision gadget a further $O(\log \kappa)$ times yields a 3-local hypergraph (after II). Then the 3-to-2 gadget finally gives a 2-local Hamiltonian interaction graph (step III).}
\label{fg step1}
\end{figure}

\subparagraph{Interaction strength}
By the triangle inequality $\norm{H_\textup{target}} = O(n\delta \mu_0)$.
The first round of perturbation theory has $\norm{H_2}<\norm{H_1} = O(\norm{H_\textup{target}})= O(n\delta \mu_0)$ since $H_1$ includes $H_\textup{else}$. 
Hence $\Lambda_1 = O(\norm{H_\textup{target}}) = O(n\delta \mu_0)$. 
For the first round of simulation, we can pick $\mu_1 \propto \Lambda_1^6 (\frac{1}{\epsilon^2} + \frac{1}{\eta^2}) $ by Lemma \ref{second order sim}.  

For the second round, the new target Hamiltonian has $O(n \delta) + O(n \delta)$ terms of strength $O(\mu_1)$, as each application of subdivision gadget adds $O(1)$ terms. 
Again $\norm{H_1}>\norm{H_2}$ giving $\Lambda_2 \propto n \delta \mu_1$ and $\mu_2 \propto \Lambda_2^6(\frac{1}{\epsilon^2} + \frac{1}{\eta^2})$. 

Since this gadget will be applied $O(\kappa)$ times (some are applied in parallel so only $O(\log\kappa)$ rounds of perturbation in total), we can upper bound the number of terms by $O(n\delta\kappa)$.
After $r$ rounds of simulation the interaction strength is
\begin{align}
    \mu_r &\propto (n \delta\kappa)^{6r+1} \mu_0^{6r} \left(\frac{1}{\epsilon^2}+ \frac{1}{\eta^2}\right)^{6(r-1)+1} \\
    &\in O\left((n\delta\kappa)^{6r + 1}\left[\mu_0\left(\frac{1}{\epsilon^2} + \frac{1}{\eta^2}\right)\right]^{6r}\right) \\
    &\in O\left(\left[\poly(n,\delta,\kappa) \mu_0 \left(\frac{1}{\epsilon^2} + \frac{1}{\eta^2} \right)  \right]^{6r} \right).
\end{align}

After the $\log(\kappa)$ rounds required to obtain a 2-local Hamiltonian, we obtain an interaction strength of
\begin{equation}
\mu\in O\left(\left[\poly(n,\delta,\kappa) \mu_0 \left(\frac{1}{\epsilon^2} + \frac{1}{\eta^2} \right)  \right]^{6\log(\kappa)} \right).
\end{equation}

\paragraph{Step 2 - Reduce degree}
The ancillas introduced in step 1 have a maximum degree\footnote{Since the ancillas introduced in subdivision gadgets have degree 2 and the 3-to-2 gadget replaces one 3-local interaction with 3, 2-local ones.} of $6$.
After step 1 the degree of the original qubits is at most $3\delta$.
We will use the triangle gadget (see \cref{sect: old pert}) to simulate $H'$ by a 2-local Hamiltonian with degree $4$ denoted $H''$.

Lay out the $O(n\delta\kappa)$ qubits of $H'$ in a line where each vertex in the graph has $O(\delta)$ incoming edges.
Subdivide each edge just once so that the vertices with degree $>2$ are isolated on the line and interact only with an ancilla with degree 2 (this requires one round of simulation and $O(n\delta\kappa)$ ancillas).
See~\cref{fg step2} step I for reference. 

Consider a given degree-$\delta$ vertex and the $O(\delta)$ ancillas it is directly connected to. 
$O(\log \delta)$ applications of the triangle gadget in parallel introduces $\log \delta$ ancillas and reduces the high degree vertex to $\delta/2$.
Recursively iterating this procedure $O(\log \delta)$ times requires $O(\delta)$ ancillas and produces a tree of depth $O(\log \delta)$ which can be placed on a square lattice, resulting in an interaction graph with maximum degree $4$.
See~\cref{fg step2} step II for reference. 

Individually executing this procedure for all high degree qubits in $H'$ requires $O(n\kappa \delta^2)$ ancillas and $O(\log \delta)$ rounds of perturbation.

\begin{figure}[h!]
\centering
\includegraphics[trim={0cm 0cm 0cm 0cm},clip,scale=0.5]{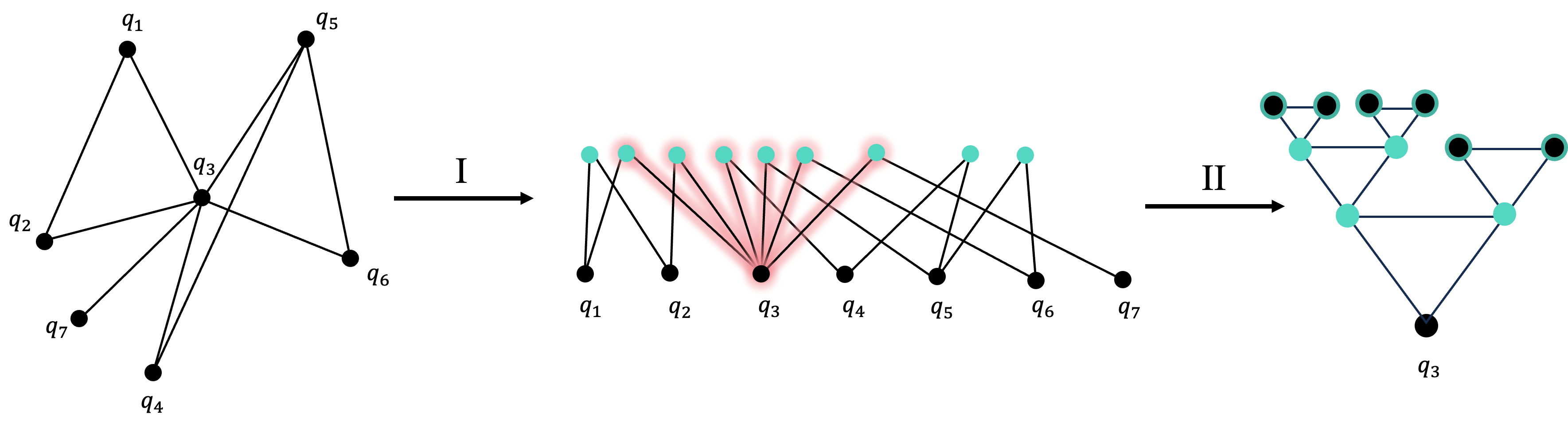}
\caption{Example of reducing the degree in an interaction graph. Step I depicts subdividing all interactions once so every original qubit only shares an edge with a new ancilla (drawn in turquoise) that has degree 2. Each high degree qubit (in black) is then separately iterated with $O(\log \delta)$ applications of the triangle gadget (\cref{triangle}) yielding a graph with maimum degree 4. In II we have shown this procedure for one qubit $q_3$ but it will occur in parallel for all qubits on the line.}
\label{fg step2}
\end{figure}

\subparagraph{Interaction strength}

After $\log(\delta)$ rounds, we will have introduced at most $O(n\delta \kappa)$ ancillas, and each has constant degree, thus the bound $\Lambda_{i+1} \propto n \delta \kappa \mu_i$ still holds, and we obtain 
\begin{equation}
  \mu \in O\left(\left[\poly(n,\delta,\kappa) \mu_0 \left(\frac{1}{\epsilon^2} + \frac{1}{\eta^2} \right)  \right]^{6(\log(\kappa) + \log(\delta))} \right).
\end{equation}

\paragraph{Step 3 - Making geometrically local}
Now lay out the Hamiltonian from Step 2 on a 2D grid of size $O(n\kappa\delta)\times O(n\kappa\delta)$ \cref{fg bigsim}.
All interactions are 2 local and the maximum degree is 4.
However the Hamiltonian graph is not geometrically local as interactions between qubits highlighted in blue may be $O(n\kappa\delta)$ long and there are many crossings.
\begin{figure}[h!]
\centering
\includegraphics[trim={0cm 0cm 0cm 0cm},clip,scale=0.5]{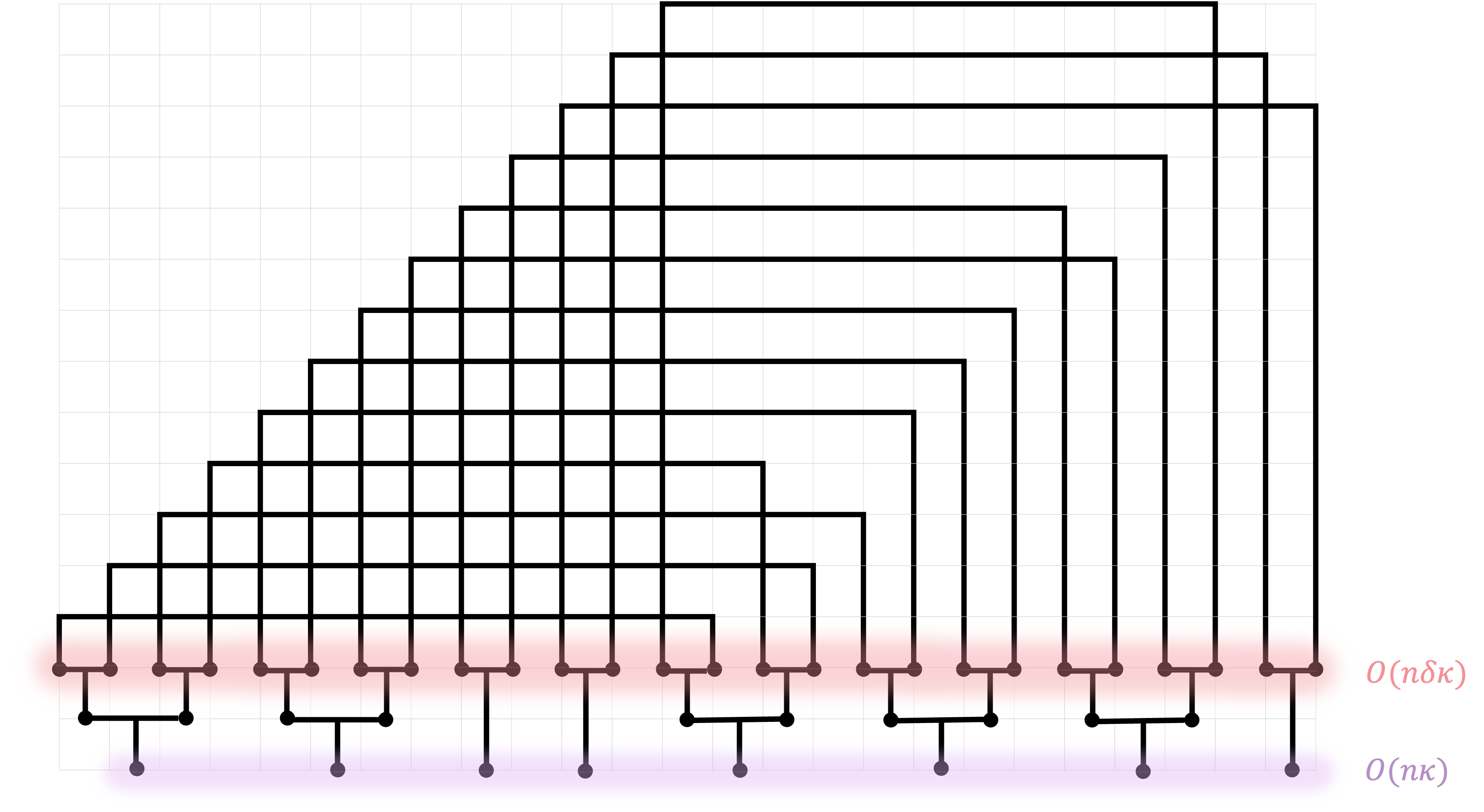}
\caption{Lay out the Hamiltonian from step 2 on a 2d lattice. The qubits highlighted in pink are the qubits from step 1. In this example the maximum degree of the original graph was 4. The graph is not geometrically local due to the interactions between the blue qubits. }
\label{fg bigsim}
\end{figure}

\paragraph{Step 3a - Remove long-range interactions}
Generally the length of the edges in the graph after step 2 will be $O(n\kappa \delta)$, and these must now be fitted onto edges of the square lattice. 

We can use one application of the newly introduced long-range gadget per long interaction which introduces the $O(n\kappa\delta)$ ancillas in a single round of perturbation. 

\subparagraph{Interaction strength}
After applying a long-range gadget to each long-range edge on the interaction graph, we obtain at most $\propto (n \kappa \delta)^2$ ancillas -- in the worst case filling up the $O(n \kappa \delta )\times O(n \kappa \delta)$ grid. 
Hence the norm of the Hamiltonian to be simulated at step $i$ is $\propto (n \delta \kappa)^2 \mu_i$. 
However, we need to verify that $\Lambda_{i+1}$ upper bounds $\norm{H_1},  \norm{H_2}$.
$\norm{H_2}$ scales as $O(\sqrt{1/C}) = O(\sqrt{n})$ from~\cref{lm bounding C}, and since $\norm{H_1}\leq \norm{H_\textup{else}} + \frac{D}{C} \propto (n\kappa\delta)^2 \mu_i  $, then $\max\{\norm{H_1}, \norm{H_2}\} \leq (n \delta \kappa)^2 \mu_i \propto \Lambda_{i+1}$.
This does little to change the scaling of $\mu$: as we only need $O(1)$ rounds of simulation, obtaining
\begin{equation}
\mu \in O\left(\left[\poly(n,\delta,\kappa) \mu_0 \left(\frac{1}{\epsilon^2} + \frac{1}{\eta^2} \right)  \right]^{6(\log(\kappa) + \log(\delta) + O(1))} \right).
\end{equation}

\paragraph{Step 3b - Remove crossings}
We have introduced a potential $O(n\kappa\delta)$ crossings per qubit at the leaves of the tree. 
Hence $O(n^2\kappa^2\delta^2)$ crossings in the whole interaction graph.

Each crossing can be removed individually by $O(n^2\kappa^2\delta^2)$ parallel applications of the crossing gadget (\cref{sect: old pert}).
This introduces $O((n\kappa\delta)^2)$ new ancillas but only requires $O(1)$ round of perturbation. 
\footnote{The interactions of the crossing gadget (\cref{fg crossing}) may be subdivided
 so that the crossing gadget fits on the square lattice. If necessary the lattice spacing can be made twice as narrow to make space to fit two crossing gadgets next to each other. This only makes a constant factor difference to the resources needed.}
\subparagraph{Interaction strength}
The argument follows identically to the previous step. We obtain
\begin{equation}
\mu \in O\left(\left[\poly(n,\delta,\kappa) \mu_0 \left(\frac{1}{\epsilon^2} + \frac{1}{\eta^2} \right)  \right]^{6(\log(\kappa) + \log(\delta) + O(1))} \right).  
\end{equation}

At the end of this procedure we obtain $H_\textup{sim}$ acting on a 2D grid of $O(n \kappa \delta )\times O(n \kappa \delta)$ qubits that simulates $H_\textup{target}$ with simulation parameters $(\mu/2,\eta,\epsilon)$. 
At each stage of simulation we use a perturbation gadget (\cref{subdivision,triangle,long-range,crossing}) which fit the definition of an approximate simulation (\cref{defn approx sim}) with $p=1$, $q=0$ and the isometry $T \ket{\psi} \mapsto \ket{\psi}\otimes \ket{\text{anc}}$.
We either apply these gadgets in parallel (see \cref{sect Ham sim} and references within) or in series requiring multiple rounds of perturbation. 

The isometry for the simulation describing all gadgets applied in parallel is again of the form $T: \ket{\psi}\mapsto \ket{\psi}\otimes \ket{\text{anc}}$ if this is the case for each individual gadget. 
This can be seen since the projector into the joint ancillary groundspace is given by $P_- = \prod_{i}P_-^{(i)}$.
All $P_-^{(i)}$ are rank one projectors acting on the subset $S_i$ and act trivially on the rest of the ancillas where $S_i$ are disjoint subsets of ancillary qubits.
Thus, $P_-$ is a rank one projector acting on all ancillas inserted at that level: $S_1\cup S_2 \cup ... \cup S_k$\footnote{Since rank of a Kronecker product are multiplicative ($\mathrm{rank}(A\otimes B) = \mathrm{rank}(A)\cdot \mathrm{rank}(B)$)}.
Similarly the isometry of the simulation describing the concatenation of two simulations is of the form $T:\ket{\psi}\mapsto \ket{\psi}\otimes \ket{\text{anc}}$ if this is the case for the component simulations. 
When we concatenate approximate simulations i.e. let A be a simulation of B and B be a simulation of C as in \cref{lm transitive} we use the composed encoding map $\mathcal{E} = \mathcal{E}_A \circ \mathcal{E}_B$.
Where $\mathcal{E}_B : \mathcal{H} \mapsto \mathcal{H}\otimes \mathcal{B}$ and $\mathcal{E}_A : (\mathcal{H}\otimes \mathcal{B}) \mapsto (\mathcal{H}\otimes \mathcal{B})\otimes\mathcal{A}$ with $\mathcal{E}_B(M) = T_BMT_B^\dagger = M \otimes P_B$ and $\mathcal{E}_A(M) = T_AMT_A^\dagger = M \otimes P_A$.
Then $\mathcal{E}_A \circ \mathcal{E}_B (M) = M \otimes P_A \otimes P_B$ and again due to multiplicative rank under Kronecker product $P_A \otimes P_B$ is rank 1 and we can describe $\mathcal{E}(M) = T M T^\dagger$ where $T: \ket{\psi} \mapsto \ket{\psi}\otimes \ket{\text{anc}}_{AB}$ for some pure state $\ket{\text{anc}}$.
\end{proof}

\subsection{Arbitrary sparse Hamiltonians}

Previously `efficient'\footnote{An `efficient' simulation in Hamiltonian complexity uses resources (ancillas and interaction strength) that scales at most polynomially in the system size.} simulations by 2D lattice Hamiltonians constructed from gadgets were accessible to \emph{spatially sparse} Hamiltonians (see Lemma 47 of \cite{Cubitt2019}) but introducing long interactions caused an exponential increase in the required interaction strengths.
\cite{Zhou} constructed efficient simulations for all local Hamiltonians using a history state method.
Our new gadget is an exponential advantage over existing gadgets for this reduction and a polynomial improvement on the ancillas required compared to the previous efficient simulation.
This facilitates a constructive simulation of \emph{arbitrary sparse} Hamiltonians on a 2D lattice using only polynomial interaction strengths and quadratically many ancillas.
Furthermore the number of ancillas in this work is independent of the simulation error. 
This is interesting from an analogue simulation perspective since general sparse Hamiltonians are now -- at lease theoretically -- accessible to restricted simulator Hamiltonians with simple connectivity.

We summarise this as a simple corollary of \cref{thm general Ham}:
\begin{corollary}[Simulating sparse Hamiltonians]\label{cor sparse}
    Given a Hamiltonian, $H = \sum_i h_i$, acting on $n$ qubits where $h_i$ is a Pauli rank one operator.
    Let be $H$ be \emph{sparse}: $O(1)$-local and the maximum degree is $O(1)$.

$\exists$ a nearest neighbour Hamiltonian acting on a 2D lattice of $N = O(n^2)$ qubits, $\tilde{H} = \sum_j \tilde{h}_j$,
that is a $(\Delta, \eta, \epsilon)$-simulation of $H$ with,
\begin{align*}
    \mu &\in O\left(\poly\left(n, \mu_0 \left(\frac{1}{\epsilon^2} + \frac{1}{\eta^2} \right)  \right) \right) \\
    \Delta &= \mu/2
\end{align*}
where $\mu = \max_j\norm{\tilde{h}_j}$, and $\mu_0 = \max_i \norm{h_i}$.
Further the encoding isometry has the form $T : \ket{\psi} \longrightarrow \ket{\psi} \otimes \ket{\textup{anc}}$ for some state of the ancillary qubits $\ket{\textup{anc}}$. 
\end{corollary}
\begin{proof}
    The proof is immediate by employing \cref{thm general Ham} setting $\delta,\kappa \in O(1)$.
\end{proof}

\section{LDPC codes in 2D}

A code can be associated with a Hamiltonian.
Given a generating set of stabilizers for a CSS code,
\begin{align}
A_r \equiv \prod_{n\in N_r} X_n & \qquad 1\leq r\leq R\\
B_s \equiv \prod_{n\in N_s} Z_n & \qquad 1\leq s\leq S,
\end{align}
where $[A_r,B_s]=0$ for all $r\in[1,R]$ and $s\in[1,S]$.
The code Hamiltonian reads,
\begin{equation}
H = -a\sum_{r=1}^R A_r - b \sum_{s=1}^S B_s,
\end{equation}
for $a,b>0$.

We begin with the Hamiltonian $H_\textup{LDPC}$, corresponding to a good LDPC code acting on $n$ qubits.
It has been shown by \cite{panteleev2021asymptotically, Breuckmann} that `good' codes exist with $O(n)$ stabilizers of constant weight.
Consequently we obtain a $\kappa$-local Hamiltonian with $\kappa = O(1)$ with respect to $n$.
Furthermore each qubit is acted on non-trivially by a constant number of the stabilizer terms and so the maximum degree of the interaction graph is $\delta = O(1)$ with respect to $n$. 
Crucially these codes do \emph{not} have so called `spatially sparse' interaction hypergraphs, since when planarised the edges are of length $O(n)$ with many crossings.
Instead they have the key property of being \emph{expander}.
Hence using previous gadget techniques to planarise these Hamiltonians would require interaction strengths and therefore energy scaling exponentially in $n$. 

Employing \cref{cor sparse} to sparse Hamiltonians describing `good' codes immediately gives us a 2D geometrically local Hamiltonian with approximately the same energy landscape using only polynomial interaction strengths.
Crucially we also need to examine the eigenstates of the simulation and show that they are closely related to the eigenstates of the code. 

We will use the Gentle Measurement Lemma to quantify how a measurement does not drastically disturb a state if the probability of a given outcome is high:
\begin{lemma}[Gentle Measurement Lemma: \cite{Winter2014}]\label{lm gentle}
    Consider a density operator $\rho$ and a measurement operator $M$ where $0\geq M \geq \idty$ (could be an element of a POVM).
    Suppose that $M$ has a large probability of detecting state $\rho$
    \[ \trace [M \rho ] \geq 1- \epsilon,\]
    where $0<\epsilon\leq 1$.
    Then the post measurement state $\rho' := \frac{\sqrt{M}\rho\sqrt{M}}{\trace[M\rho]}$ is close to the original state,
    \[\norm{\rho - \rho'}_1 \leq 2\sqrt{\epsilon}.\]
\end{lemma}

In the following result we show two complementary conditions satisfied by simulations where the isometry is particularly simple (of the form $T:\ket{\psi}\mapsto \ket{\psi}\otimes \ket{\textup{anc}}$) -- as is the case for all the gadgets used in \cref{thm general Ham}. 
The first is that for any low energy state in the simulation subspace, their energy as measured by the simulation Hamiltonian is close to what would be measured by the target Hamiltonian.
In particular, we would like to ensure that the ground states of $H_\textup{sim}$ are almost ground states of $H_\textup{target}$, hence "soundness".
Conversely, we would also like to ensure that \emph{all} ground states of $H_\textup{target}$ are indeed also ground states of $H_\textup{sim}$, hence "completeness".

\begin{lemma}\label{lm complete and sound}
    Let $H_s$ be a ($\Delta,\epsilon,\eta$)- simulation of the target Hamiltonian $H_t$.
    The local Encoding $\mathcal{E}: \mathcal{H}\mapsto \mathcal{H}\otimes \mathcal{A}$ is described by $\mathcal{E}(M) = T M T^\dagger = M \otimes P_A$ where $M\in \mathcal{H}$ and $P_A \in\mathcal{A}$ and $P_A$ is a rank one projector.
    The general encoding mapping into the low-energy subspace is given by $\tilde{\mathcal{E}}(M) = \tilde{T}M\tilde{T}^\dagger$ with $\norm{T - \tilde{T}}\leq \eta$(see~\cref{defn approx sim}).
    
    Then, for all sufficiently small $\eta$, we have
    \begin{enumerate}
        \item Let $\rho\in \mathcal{S}(\mathcal{H}\otimes \mathcal{A})$ be low energy energy state $\trace(\rho H_s)\leq \epsilon'$ (with $\epsilon'\ll \Delta$), then \[ \trace(\trace_{\mathcal{A}}(\rho)H_t)\leq 5\epsilon' + \epsilon + O_\downarrow\left(\sqrt{\eta}\right)\norm{H_t}.\]
        A condition we call \textbf{`soundness'}. 
        
        \item For any state $\sigma \in \mathcal{S}(\mathcal{\hbt})$, there exists a state $\tilde{\sigma} \in \mathcal{S}(\mathcal{H}\otimes \mathcal{A})$, such that
        \[
         \norm{\tr_\mathcal{A}(\tilde{\sigma}) - \sigma}_1 \in O_\downarrow(\eta),
        \] 
        and
        \[
        \abs{\tr(\tilde{\sigma} H_s) -  \tr(\sigma H_t) } \leq \epsilon  .
        \]
        A condition we call \textbf{`completeness'}. 
    \end{enumerate}
    Where $g(x)\in O_\downarrow (f(x))$ if there exists $c,$ $\delta$ such that for all $0<\abs{x}<\delta$ we have $\abs{g(x)}\leq c \cdot f(x)$.
    $\mathcal{S}(\mathcal{H})$ denotes the states in the Hilbert space $\mathcal{H}$.
\end{lemma}

\begin{proof}
For the sake of clarity, we will use the notation $g(x) \in O_\downarrow (f(x))$ since we are interested in the small $\eta$, $\epsilon$ regime. 

We prove the two statements in turn,
\paragraph{Proof of soundness}
Denote by $\mathcal{D}$ the subset of states in $\mathcal{S}(\mathcal{H} \otimes \mathcal{A})$ that are in $\mathrm{supp}(P_{\leq \Delta})$.
    Then $\rho = (1-p)\rho_1 + p\rho_2$ where $p\in [0,1]$ and $\rho_1\in\mathcal{D}$, $\rho_2\in\mathcal{S}(\mathcal{H} \otimes \mathcal{A})\setminus \mathcal{D}$.
    Due to the definition of simulation the energy of $\rho_2$ is lower bounded by $\trace(\rho_2 H_s) \geq \Delta$.

    Since $\rho$ is a low energy state the probability of being in $\mathcal{S}(\mathcal{H} \otimes \mathcal{A})\setminus \mathcal{D}$ is then necessarily small, $p\leq \frac{\epsilon'}{\Delta}$.
    Therefore $\exists$ a state $\phi \in \mathcal{D}$ that is close to $\rho$ in trace distance:
    \begin{equation}\label{eqn rho phi}
        \norm{\rho - \phi}_1 \leq \frac{2\epsilon'}{\Delta}.
    \end{equation}

    Given such a state $\phi$:
    \begin{align}
        \trace(TT^\dagger \phi T T^\dagger)  = \trace(\tilde{T}\tilde{T}^\dagger\phi \tilde{T}^\dagger \tilde{T}) + \trace(TT^\dagger \phi T^\dagger T - \tilde{T}\tilde{T}^\dagger\phi \tilde{T}^\dagger \tilde{T}).
    \end{align}
    We have $\tilde{T}\tilde{T}^\dagger \phi \tilde{T}\tilde{T}^\dagger = \phi$, since $\tilde{T}\tilde{T}^\dagger = \tilde{\mathcal{E}}(\idty) = P_{\leq \Delta}$, therefore $\trace(\tilde{T}\tilde{T}^\dagger\phi \tilde{T}^\dagger \tilde{T}) =1$.
    Since $\norm{\tilde{T} - T} \leq \eta$, using the triangle inequality we can see that that $\norm{\tilde{T}\tilde{T}^\dagger - {T}{T}^\dagger} \leq 2\eta + \eta^2$.
    This gives 
    \begin{equation}
        \tr({T}{T}^\dagger \phi {T}{T}^\dagger) \geq 1 - r,
    \end{equation}
    where $r = 2\eta + \eta^2$.
    Finally by the Gentle Measurement \cref{lm gentle},
    \begin{equation}
       \norm{\phi - \frac{{T}{T}^\dagger \phi {T}{T}^\dagger }{\tr({T}{T}^\dagger \phi {T}{T}^\dagger)}}_1 \leq 2 \sqrt{r}.
    \end{equation}
    Note that $\frac{{T}{T}^\dagger \phi {T}{T}^\dagger }{\tr({T}{T}^\dagger \phi {T}{T}^\dagger)} = \tau \otimes P_A$. 
    Let $\phi = \tau \otimes P_A + \delta$, and $\tr_\mathcal{A}(\phi) = \tau + \delta'$, since the trace distance decreases under CPTP maps, we get $\norm{\delta'}_1 \leq \norm{\delta}_1 \leq 2\sqrt{r}$.
    
    Now, write $H_s|_{\leq \Delta} = {T}H_t{T}^\dagger + H_A$, and claim that $\norm{H_A} \leq \epsilon + 2\eta\norm{H_t}$ since 
    \begin{align}
        \norm{H_s|_{\leq \Delta} - {T}H_t{T}^\dagger} &= \norm{H_s|_{\leq \Delta} - \mathcal{E}(H_t) + \tilde{\mathcal{E}}(H_t) - \tilde{\mathcal{E}}(H_t)}\\
        & \leq \norm{H_s|_{\leq \Delta} - \tilde{\mathcal{E}}(H_t)} + \norm{-\mathcal{E}(H_t) + \tilde{\mathcal{E}}(H_t)}\\
        & \leq \epsilon + 2\eta\norm{H_t},
    \end{align}
    where the final line uses that $\norm{TH_tT^\dagger - \tilde{T}H_t\tilde{T}^\dagger}\leq (\norm{T} + \norm{\tilde{T}})\norm{T - \tilde{T}}\norm{H_t}$ (Lemma 18 from \cite{Cubitt2019}) and $\norm{H_s|_{\leq \Delta} - \tilde{\mathcal{E}}(H_t)}\leq \epsilon$ by definition.

    We are now in a position to prove the first statement,
    \begin{align}
        \abs{\trace(\rho H_s) - \trace(\trace_\mathcal{A}(\rho)H_t)} &= \abs{\trace((\rho-\phi)H_s) + \trace(\phi(H_s))- \trace(\trace_\mathcal{A}(\phi)H_t) + \trace((\trace_\mathcal{A}(\phi)- \trace_\mathcal{A}(\rho))H_t)}\\
        & \leq 2\epsilon' + \frac{2\epsilon'\norm{H_t}}{\Delta} + \abs{\trace(\phi(H_s))- \trace(\trace_\mathcal{A}(\phi)H_t))}\label{eqn res 1}.
    \end{align}
    We can upper bound, 
    \begin{align}
        \abs{\trace(\phi(H_s))- \trace(\trace_\mathcal{A}(\phi)H_t))} & = \abs{\tr(\phi (T H_t T^\dagger + H_A)) - \tr((\tau + \delta') H_t)} \\
        & = \abs{\tr((\tau\otimes P_A + \delta) (H_t \otimes P_A)) + \tr(\phi H_A)) - \tr(\tau H_t) - \tr(\delta' H_t)} \\
        &  = \abs{\tr(\tau H_t) + \tr(\delta (H_t \otimes P_A)) + \tr(\phi H_A)) - \tr(\tau H_t) - \tr(\delta' H_t)} \label{eqn P_A 2}\\
        & \leq \norm{\delta}_1\norm{H_t \otimes P_A} + \norm{\phi}_1\norm{H_A} + \norm{\delta'}_1\norm{H_t} \\
        & \leq  2\sqrt{r}\norm{H_t} + \epsilon + 2\eta \norm{H_t} + 2\sqrt{r}\norm{H_t} \\
        & \leq \epsilon + (4\sqrt{r} +2\eta) \norm{H_t}\label{eqn res 2},
    \end{align}
    where \cref{eqn P_A 2} uses that $\tr(P_A^2) = 1$ since it is a rank one projector.

    Combining \cref{eqn res 1} and \cref{eqn res 2} gives the result. 

\paragraph{Proof of completeness}
    For any state $\sigma \in \mathcal{S}(\hbt)$
\begin{align}
    \tr(\tilde{T}\sigma \tilde{T}^\dagger H_s) &= \tr(\tilde{T}\sigma \tilde{T}^\dagger H_s|_{\leq \Delta}) \\
    &= \tr(\tilde{T}\sigma \tilde{T}^\dagger \tilde{T} H_t \tilde{T}^\dagger) + \tr(\tilde{T}\sigma \tilde{T}^\dagger H_B) \\
    &= \tr(\sigma H_t) + \tr(\tilde{T}\sigma \tilde{T}^\dagger H_B)
\end{align}
where $H_s|_{\leq \Delta} = \tilde{T}H_t\tilde{T}^\dagger + H_B$, where $\norm{H_B} \leq \epsilon$ by the definition of a $(\Delta, \epsilon, \eta)$-simulation
We conclude that $\abs{\tr(\tilde{T}\sigma \tilde{T}^\dagger H_s) -  \tr(\sigma H_t) } \leq \epsilon$ .

Write $\tilde{\sigma} = \tilde{T}\sigma\tilde{T}^\dagger $, then we also have

\begin{align}
    \norm{\tr_\mathcal{A}(\tilde{\sigma}) - \sigma}_1 &\leq \norm{\tilde{\sigma} - \sigma\otimes P_A}_1 \\
    &=  \norm{\tilde{T}\sigma \tilde{T}^\dagger - T \sigma T^\dagger}_1 \\
    &\in O_\downarrow(\eta).
\end{align}
\end{proof}

We can now summarise our main result concerning simulating LDPC codes in the following theorem.

\begin{theorem}[Simulating LDPC codes]\label{thm LDPC}
    Given a Pauli Hamiltonian $H = \sum_i h_i$ corresponding to an arbitrary LDPC code $\mathcal{C}$ acting on $n$ qubits, $\mathcal{H} = (\mathbb{C}^2)^{\otimes n}$.
    There $\exists$ a 2D nearest-neighbour Hamiltonian, $\tilde{H} = \sum_j \tilde{h}_j$, acting on a lattice of $N\in O(n^2)$ qubits, $\mathcal{H}\otimes \mathcal{A} = (\mathbb{C}^2)^{\otimes N}$. 
    $\tilde{H}$ is a $(\Delta,\eta,\epsilon)$-simulation of $H$ such that,
    \begin{enumerate}
        \item The interaction strength of $\tilde{H}$, denoted $\tilde{\mu}$, scales polynomially with $n$:
        \[\mu \in O\left(\poly\left(n, \mu_0 \left(\frac{1}{\epsilon^2} + \frac{1}{\eta^2} \right)  \right) \right)\]
        where $\mu_0$ is the interaction strength of $H$.
        \item The low-energy eigenspectrum of $\tilde{H}$ for eigenvalues $<\Delta$, for $\Delta=\mu/2$ is a controllable approximation of $H$:
        \[\abs{\lambda_i(H) - \lambda_i(\tilde{H})}\leq \epsilon,\]
        where we denote by $\lambda_i(H)$ the $i^\text{th}$ smallest eigenvalue of $H$.
        \item The eigenvectors corresponding to the low energy subspace of $\tilde{H}$ (again $<\Delta$) are close to the corresponding eigenvectors of $H$. For sufficiently small 
        \begin{enumerate}
            \item Soundness: Let $\rho\in \mathcal{S}(\mathcal{H}\otimes \mathcal{A})$ be low energy energy state $\trace(\rho H_s)\leq \epsilon'$, then \[\trace(\trace_{\mathcal{A}}(\rho)H_t)\leq 5\epsilon' + \epsilon + O_\downarrow\left(\sqrt{\eta}\right)\norm{H_t}.\]
            \item Completeness: For any state $\sigma \in \mathcal{S}(\mathcal{\hbt})$, there exists a state $\tilde{\sigma} \in \mathcal{S}(\mathcal{H}\otimes \mathcal{A})$, such that
        \[
         \norm{\tr_\mathcal{A}(\tilde{\sigma}) - \sigma}_1 \in O_\downarrow(\eta),
        \] 
        and
        \[
        \abs{\tr(\tilde{\sigma} H_s) -  \tr(\sigma H_t) } \leq \epsilon.
        \]
        \end{enumerate}
        Where we denote by $\mathcal{S}(\mathcal{H})$ the set of states in the Hilbert space $\mathcal{H}$ and $g(x)\in O_\downarrow (f(x))$ if there exists $c,$ $\delta$ such that for all $0<\abs{x}<\delta$ we have $\abs{g(x)}\leq c \cdot f(x)$.
     \end{enumerate}
\end{theorem}
\begin{proof}
In order to be LDPC the interaction graph of $H$ has locality $k\in O(1)$, maximum degree $\delta \in O(1)$ and is therefore sparse.
We can thus use \cref{cor sparse} to construct a 2D nearest neighbour $(\Delta,\eta,\epsilon)$-simulating Hamiltonian with interaction strengths as quoted in point (1).

For all the gadgets used $p=1$ and $q=0$, hence as a direct result of point (i) from \cref{physical-properties} we find the low energy eigenspectrum is approximately preserved -- point (2).

To prove (3) we use \cref{lm complete and sound} and the form of the isometry given for the encoding in \cref{cor sparse} is $T: \ket{\psi} \mapsto \ket{\psi}\otimes \ket{\text{anc}}$ and hence $\mathcal{E}(H) = T H T^\dagger = H \otimes P_A$ where $P_A$ is a rank one projector.
Finally to demonstrate the eigenvector relation we employ \cref{lm complete and sound} which immediately gives the result. 
\end{proof}

Additionally we can consider applying \cref{lm complete and sound} to \cite{Zhou}.
In the history state construction the isometry mapping from the target Hilbert space to the enlarged simulator Hilbert space can be made close to the form $V\ket{\psi} \mapsto \ket{\psi}\otimes \ket{\text{anc}}$ by increasing the idling time.
Therefore, the simulation obeys the constraints of \cref{lm complete and sound}.
This would yield a result similar in spirit to \cref{thm LDPC} with the key difference that the simulator Hamiltonian acts on $O(\mathrm{poly}(n,\epsilon^{-1}))$ qubits.
The key advantage of using our perturbative protocol is the number of ancillas involved: the polynomial scaling has bounded degree 2 and is independent of the simulation error.
Having a single lattice and Hamiltonian structure for different error tolerances is practially desirable for this application. 
We leave considering whether the techniques of \cite{Zhou} could be optimised for codes to future work. 

\section*{Acknowledgements}
The authors would like to thank Tamara Kohler for useful discussions and feedback on a draft manuscript, as well as Arkin Tikku for innumerable insightful comments and an inordinate willingness to answer the authors' questions.
In addition the 2022 IBM QEC summer school and Coogee 2023 for fostering the collaboration.

H.\,A. ~is supported by EPSRC DTP Grant Reference: EP/N509577/1 and EP/T517793/1.
N.\,B. is supported by the Australian Research Council via the Centre of Excellence in Engineered Quantum Systems (EQUS) project number CE170100009, and by the Sydney Quantum Academy.

\newpage
\begin{appendix}
\huge
\noindent\textbf{Appendices}
\normalsize

\section{Other perturbative gadgets}\label{sect: old pert}

In this section we give as reference proofs of the perturbation gadgets used in \cref{thm general Ham} that were taken directly from the literature.
They all first appeared in \cite{Oliveira2008}, however the phrasing and proof technique given here follows \cite{Kohler2019}.
The mechanism of these simpler gadgets is a useful insight into how the long-range gadget is constructed.

\begin{figure}[h!]
\centering
\includegraphics[trim={0cm 0cm 0cm 0cm},clip,scale=0.6]{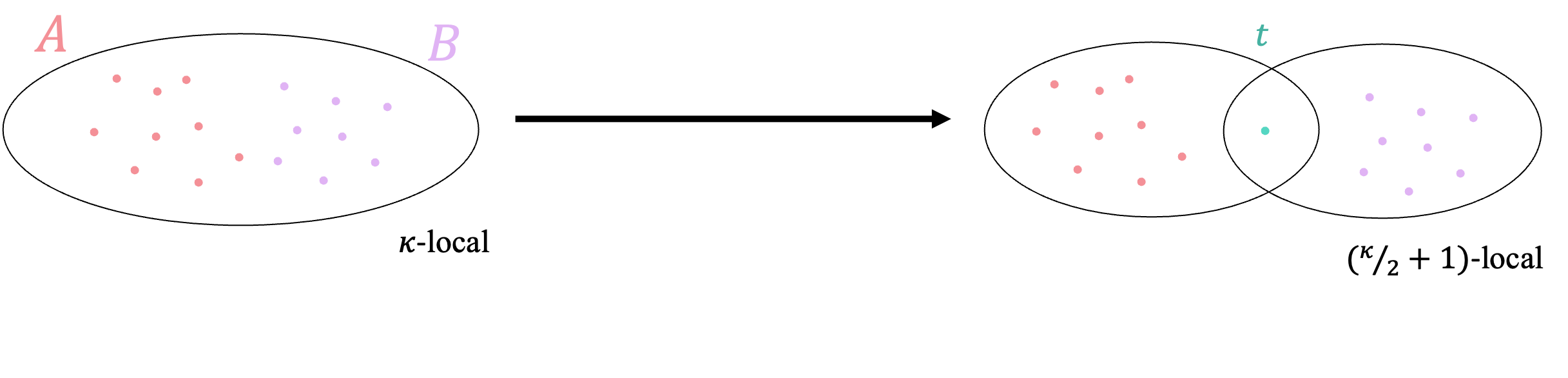}
\caption{Subdivision gadget}
\label{fg subdivision}
\end{figure}

\begin{lemma}[Qubit subdivision gadget~\cite{Kohler2019}]\label{subdivision}
The $\kappa$-local Hamiltonian,
\begin{equation}
H_\textup{target} = H_\textup{else} + P_A\otimes P_B,
\end{equation}
is $(\Delta/2,\eta,\epsilon)$-simulated by a $(\floor{\kappa/2}+1)$-local Hamiltonian, $\tilde{H}=H+V$ where,
\begin{align}
&H = \Delta \Pi_+\\
&V = H_1 + \sqrt{\Delta}H_2\\
&H_1 = H_\textup{else}+\frac{1}{2}\left(P_A^2 + P_B^2 \right)\\
&H_2 =\frac{1}{\sqrt{2}}\left( P_A\otimes X_t - P_B\otimes X_t\right).
\end{align}
Hilbert space decomposition is given by,
\begin{align}
&\Pi_- = \ket{0}\bra{0}_t\\
&\Pi_+ = \ket{1}\bra{1}_t.
\end{align}
\end{lemma}

\begin{proof}
Using the projectors given,
\begin{align}
&H_{1--} = \left(H_\textup{else}+ 2 \idty \right)\otimes \ket{0}\bra{0}_t\\
&H_{2-+} =H_{2+-} =\frac{1}{\sqrt{2}}\left( P_A - P_B\right)\otimes \ket{0}\bra{1}_t.
\end{align}
Therefore,
\begin{align}
H_{2-+}H_0^{-1}H_{2+-} & = \frac{1}{2}(P_A^2-P_B^2)\otimes \ket{0}\bra{0}_t\\
& = \left(\frac{1}{2}P_A^2 - P_A\otimes P_B + \frac{1}{2}P_B^2 \right)\otimes \ket{0}\bra{0}_t
\end{align}
Defining the isometry $T\ket{\psi}_{AB} = \ket{\psi}_{AB}\ket{0}_t$,
\begin{equation}
\norm{T H_\textup{target}T^\dagger - H_{1--} + H_{2-+}(H_{0++})^{-1}H_{2+-}}_\infty = 0.
\end{equation}
Therefore by~\cref{second order sim} $\tilde{H}$ simulated $H_\textup{target}$ for any $\epsilon>0$ given an appropriate choice of $\Delta$.
\end{proof}

\begin{figure}[h!]
\centering
\includegraphics[trim={0cm 0cm 0cm 0cm},clip,scale=0.6]{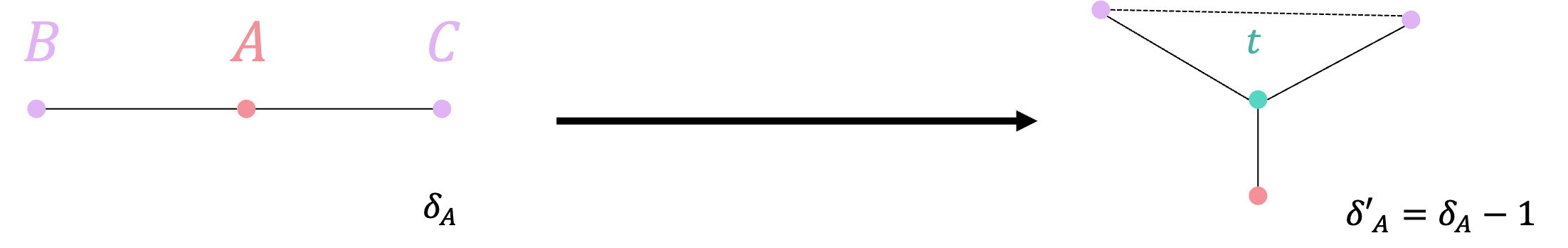}
\caption{Triangle gadget}
\label{fg fork}
\end{figure}

\begin{lemma}[Triangle qubit gadget~\cite{Kohler2019}]\label{triangle}
The Hamiltonian where qubit $A$ has degree $\delta$ in the interaction graph,
\begin{equation}
H_\textup{target} = H_\textup{else} + \alpha_{ab}P_A\otimes P_B  + \alpha_{ac}P_A\otimes ,
\end{equation}
is $(\Delta/2,\eta,\epsilon)$-simulated by a Hamiltonian, $\tilde{H}=H+V$, where qubit $A$ has degree $(\delta-1)$.
\begin{align}
&H = \Delta \Pi_+\\
&V = H_1 + \sqrt{\Delta}H_2\\
&H_1 = H_\textup{else}+\alpha_{ab}\alpha_{bc}P_B\otimes P_C^\dagger+ \frac{1}{2}\left(1 + \alpha_{ab}^2 + \alpha_{bc}^2 \right)\idty\\
& \begin{multlined}
H_2 =\frac{1}{\sqrt{2}}\left( -P_A + \alpha_{ab} P_B +\alpha_{ac} P_C\right)\otimes X_t^\dagger,
\end{multlined}
\end{align}
where $X$ is the qubit Pauli X operator.
Hilbert space decomposition is again given by,
\begin{align}
&\Pi_- = \ket{0}\bra{0}_t\\
&\Pi_+ = \ket{1}\bra{1}_t.
\end{align}
\end{lemma}

\begin{proof}
Again using the projectors given,
\begin{align}
&H_{1--} = H_1 \otimes \ket{0}\bra{0}_t\\
&\begin{multlined}
H_{2-+} =\frac{1}{\sqrt{2}}\left( -P_A +\alpha_{ab}P_B + \alpha_{ac}P_C\right)\otimes \ket{0}\bra{p-1}_t.
\end{multlined}
\end{align}
Therefore,
\begin{equation}
\begin{multlined}
(H_2)_{-+}(H_0)^{-1}(H_2)_{+-} = \frac{1}{2}(P_A^2 + \alpha_{ab}^2P_B^2 + \alpha_{bc}^2P_C^2) - 2\alpha_{ab}P_A\otimes P_B 
 - 2\alpha_{ac}P_A\otimes P_C\\
 + 2\alpha_{ab}\alpha_{bc}P_B\otimes P_C.
\end{multlined}
\end{equation}
Defining the isometry $T\ket{\psi}_{AB} = \ket{\psi}_{AB}\ket{0}_t$,
\begin{equation}
\norm{T H_\textup{target}T^\dagger - H_{1--} + H_{2-+}(H_{0++})^{-1}H_{2+-}}_\infty = 0.
\end{equation}
Therefore by~\cref{second order sim} $\tilde{H}$ simulated $H_\textup{target}$ for any $\epsilon>0$ given an appropriate choice of $\Delta$.
\end{proof}

\begin{figure}[h!]
\centering
\includegraphics[trim={0cm 0cm 0cm 0cm},clip,scale=0.6]{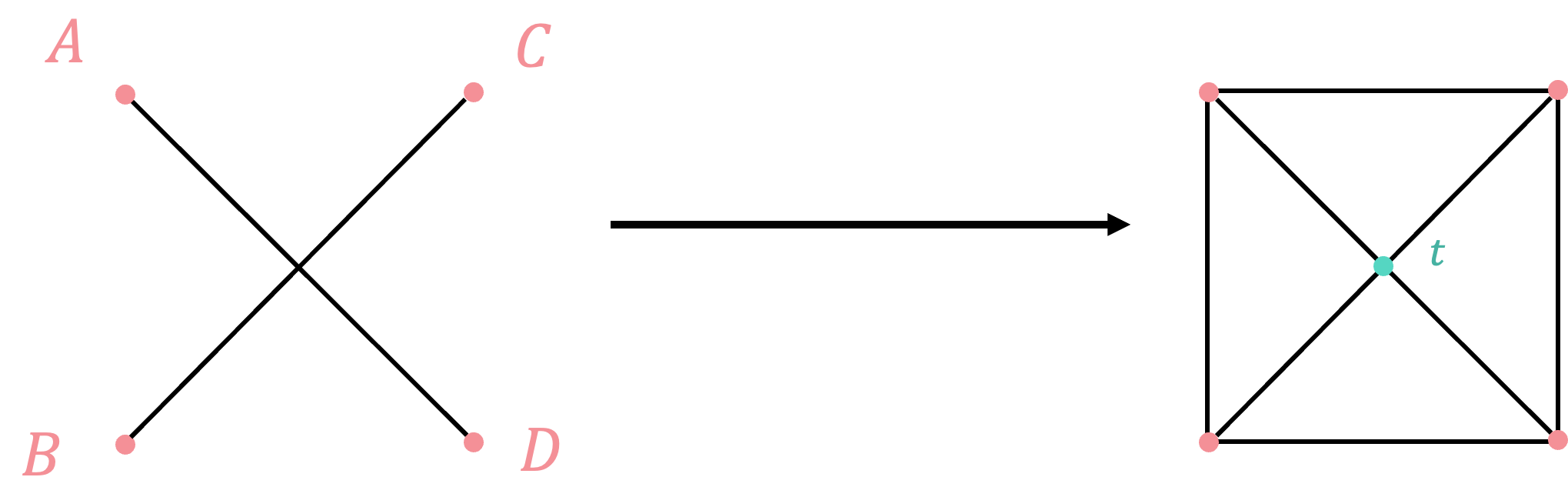}
\caption{Crossing gadget}
\label{fg crossing}
\end{figure}

\begin{lemma}[Qubit crossing gadget~\cite{Oliveira2008}]\label{crossing}
The `crossed' Hamiltonian,
\begin{equation}
H_\textup{target} = H_\textup{else} + \alpha_{ad}P_A\otimes P_D + \alpha_{bc}P_B\otimes P_C,
\end{equation}
is $(\Delta/2,\eta,\epsilon)$-simulated by a geometrically local Hamiltonian, $\tilde{H}=H+V$ where,
\begin{align}
&H = \Delta \Pi_+\\
&V = H_1 + \sqrt{\Delta}H_2\\
&\begin{multlined}
    H_1 = H_\textup{else}+\frac{1}{2}\left(\alpha_{ad}P_A^2 + \alpha_{bc}P_B^2 +P_C^2 + P_D^2\right) - \alpha_{ad}\alpha_{bc}P_A\otimes P_B + \alpha_{ad}P_A\otimes P_C \\
    + \alpha_{bc}P_B\otimes P_D - P_C\otimes P_D
\end{multlined}\\
& H_2 =\frac{1}{\sqrt{2}}\left( -\alpha_{ad}P_A\otimes X_t - \alpha_{bc}P_B\otimes X_t + P_C\otimes X_t + P_D \otimes X_t\right).
\end{align}
Hilbert space decomposition is given by,
\begin{align}
&\Pi_- = \ket{0}\bra{0}_t\\
&\Pi_+ = \ket{1}\bra{1}_t.
\end{align}
\end{lemma}

\begin{proof}
    This is a second order simulation.
    \begin{equation}
        (H_2)_{-+} = \frac{1}{\sqrt{2}}\left[-\alpha_{ad}P_A\otimes - \alpha_{bc}P_B + P_C + P_D \right]\ket{0}\bra{1}_t.
    \end{equation}
    Therefore,
    \begin{equation}
    \begin{multlined}
        (H_2)_{-+}(H_0)^{-1}(H_2)_{+-} = \frac{1}{2}\left[ \alpha_{ad}^2P_A^2 + \alpha_{bc}^2P_B^2 + P_C^2 + P_D^2 + 2\alpha_{ad}\alpha_{bc}P_A\otimes P_B - 2\alpha_{ad}P_A\otimes P_C \right. \\
        \left.- 2 \alpha_{ad}P_A\otimes P_D - 2\alpha_{bc}P_B\otimes P_C - 2\alpha_{bc}P_B\otimes P_D + 2 P_c\otimes P_D\right] \ket{0}\bra{0}_t.
    \end{multlined}
    \end{equation}
Defining the isometry $T: \ket{\psi}_{ABCD}\mapsto \ket{\psi}_{ABCD}\otimes \ket{0}_t$ then,
\begin{equation}
\norm{T H_\textup{target}T^\dagger - H_{1--} + H_{2-+}(H_{0++})^{-1}H_{2+-}}_\infty = 0.
\end{equation}
Therefore by~\cref{second order sim} $\tilde{H}$ simulated $H_\textup{target}$ for any $\epsilon>0$ given an appropriate choice of $\Delta$.
\end{proof}

\begin{lemma}[Third order simulation~\cite{Bravyi2014}]\label{third order sim}
Let $V = H_1 + \Delta^{\frac{1}{3}}H_1' + \Delta^{\frac{2}{3}}H_2$ be a perturbation action on the same space as $H_0$ such that $\max \{\norm{H_1}, \norm{H_1'}, \norm{H_2} \}\leq \Lambda$; $H_1$ and $H_1'$ are block diagonal with respect to the split $\mathcal{H}=\mathcal{H}_- \oplus \mathcal{H}_+$ and $(H_2)_{--} = 0$.
Suppose there exists an isometry $T$ such that $\Ima(T) = \mathcal{H}_-$ and:
\begin{equation}\label{eqn 3 order}
    \norm{TH_\textup{target}T^\dagger - (H_1)_{--} + (H_2)_{-+}H_0^{-1}(H_2)_{++}H_0^{-1}(H_2)_{+-}}_\infty \leq \frac{\epsilon}{2}
\end{equation}
and also that:
\begin{equation}
    (H_1')_{--} = (H_2)_{_+}H_0^{-1}(H_2)_{+-}.
\end{equation}
Then $\tilde{H}$ is a $\left(\frac{\Delta}{2}, \eta, \epsilon \right)$-simulation of $H_\textup{target}$ provided that $\Delta \geq O\left(\frac{\Lambda^{12}}{\epsilon^3}+ \frac{\Lambda^3}{\eta^3} \right)$.
\end{lemma}

\begin{figure}[h!]
\centering
\includegraphics[trim={0cm 0cm 0cm 0cm},clip,scale=0.6]{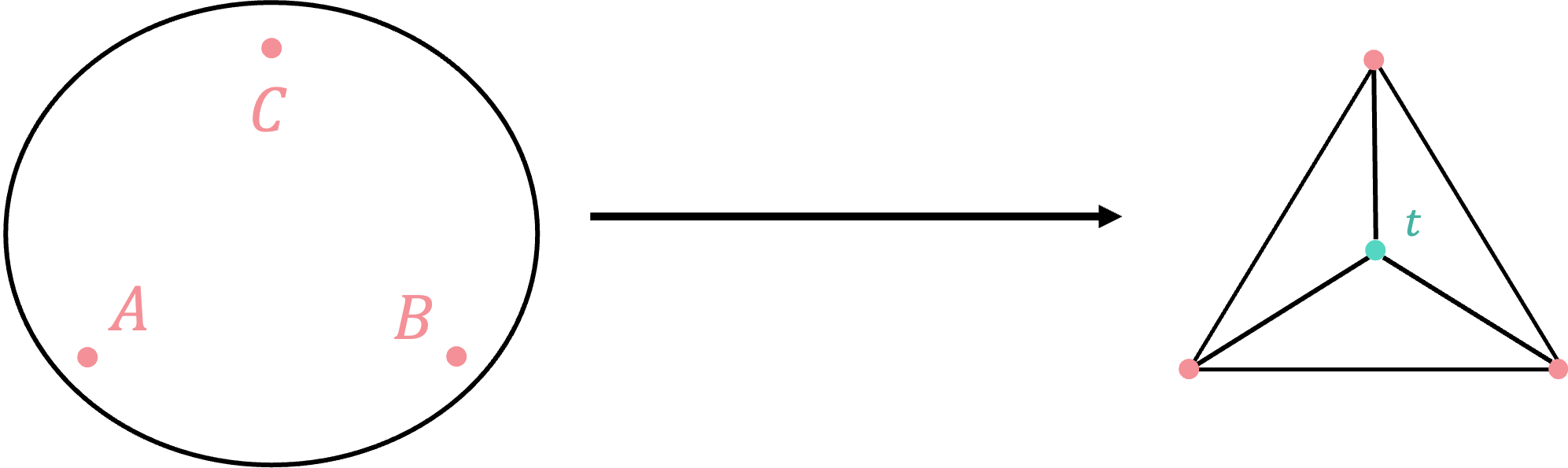}
\caption{3-to-2 qubit gadget.}
\label{fg 3-to-2}
\end{figure}

\begin{lemma}[Qubit 3-to-2 gadget~\cite{Oliveira2008}]\label{3-to-2}
The $3$-local Hamiltonian,
\begin{equation}
H_\textup{target} = H_\textup{else} + P_A\otimes P_B \otimes P_C,
\end{equation}
is $(\Delta/2,\eta,\epsilon)$-simulated by a $2$-local Hamiltonian, $\tilde{H}=H+V$ where,
\begin{align}
&H = \Delta \Pi_+\\
&V = H_1 + \Delta^{1/3}H_1' + \Delta^{2/3}H_2\\
&H_1 = H_\textup{else}+\frac{1}{2}\left(P_A^2 + P_B^2 \right)\otimes P_C\\
&H_1' = \frac{1}{2}\left(-P_A + P_B \right)^2\\
& H_2 =P_C\otimes \ket{1}\bra{1}_t + \frac{1}{\sqrt{2}}\left( -P_A\otimes X_t + P_B\otimes X_t\right).
\end{align}
Hilbert space decomposition is given by,
\begin{align}
&\Pi_- = \ket{0}\bra{0}_t\\
&\Pi_+ = \ket{1}\bra{1}_t.
\end{align}
\end{lemma}

\begin{proof}
    This is a third order simulation so to prove the above we need to demonstrate the construction satisfies the two conditions of \cref{third order sim}.

    \begin{equation}
        (H_2)_{-+} = \frac{1}{\sqrt{2}}\left[(-P_A+P_B)\otimes \ket{0}\bra{1}_t \right],
    \end{equation}
    Therefore,
    \begin{equation}
        (H_2)_{-+}H_)^{-1}(H_2)_{+-} = \frac{1}{2}\left[(-P_A+P_B)^2\otimes \ket{0}\bra{0}_t \right] = (H_1')_{--}.
    \end{equation}

    Then, 
    \begin{equation}
        (H_2)_{++}  = P_C \otimes \ket{1}\bra{1}_t,
    \end{equation}
    so that, 
    \begin{align}
        (H_2)_{-+}H_0^{-1}(H_2)_{++}H_0^{-1}(H_2)_{+-} & = \frac{1}{2}\left[(-P_A+P_B)\otimes \ket{0}\bra{1}_w \right]\left[P_C \otimes \ket{1}\bra{1}_w \right]\left[(-P_A+P_B)\otimes \ket{1}\bra{0} \right]\\
        & = \left[-P_A\otimes P_B \otimes P_C + \frac{1}{2}(P_A^2 + P_B^2)\otimes P_C \right] \otimes \ket{0}\bra{0}_t.
    \end{align}
    Defining the isometry $T\ket{\psi}_{ABC} = \ket{\psi}_{ABC}\ket{0}_t$, \cref{eqn 3 order} vanishes,
    \begin{equation}
        \norm{TH_\textup{target}T^\dagger - (H_1)_{--} + (H_2)_{-+}H_0^{-1}(H_2)_{++}H_0^{-1}(H_2)_{+-}}_\infty = 0.
    \end{equation}
    Therefore by~\cref{third order sim} $\tilde{H}$ simulated $H_\textup{target}$ for any $\epsilon>0$ given an appropriate choice of $\Delta$.
\end{proof}

\section{Impossibility of a gapped Hamiltonian}
\label{sec:no-go}
This section considers the extent to which our result can be improved using similar methods. In particular, one could 
\begin{enumerate}
    \item Substitute $X_i, X_j$ for different operators $\Gamma_i, \Gamma_j$ that potentially act on multiple qubits
    \item Pick a better ancilla state than $\ket{W}$, or a better parent Hamiltonian for $H_W$, or a combination of both.
\end{enumerate}
Unfortunately, we will show that even given this flexibility, the interaction strength of the simulation Hamiltonian of the long range gadget has to scale like $\widetilde{\Omega}(n)$.

Henceforth consider $\Gamma_i, \Gamma_j$ operators acting on a disc of radius $O(1)$. Instead of $\ket{W}$, we focus on $\ket{\psi}$, the unique ground state of an arbitrary Hamiltonian $H_\psi$ with spectral gap at least $1$.
For simplicity, we take $H_{\textup{else}} = 0$, and following the method of \cref{sect long range} we assume that the simulation Hamiltonian obeys :
\begin{equation}
    H_1 = \frac{D'}{C'}(P_A^2 \otimes P_B^2)
\end{equation}
\begin{equation}
    H_2 = \frac{1}{\sqrt{C'}}(P_A \otimes \Gamma_i - P_B \otimes \Gamma_j),
\end{equation}
with
\begin{align}
    D' &= \bra{\psi } \Gamma_i (H_{\psi ++ })^{-1} \Gamma_i \ket{\psi} + \bra{\psi } \Gamma_j (H_{\psi ++ })^{-1} \Gamma_j \ket{\psi}\\
    C' &=  \bra{\psi } \Gamma_i (H_{\psi ++ })^{-1} \Gamma_j \ket{\psi} + \bra{\psi } \Gamma_j (H_{\psi ++ })^{-1} \Gamma_i \ket{\psi}.
\end{align}

It would technically be possible to pick different pre-factors for $H_1, H_2$ but the end result would not be significantly affected.
Similarly to \cref{long-range}, we define $\Pi_- = \dyad{\psi}$, $\Pi_+ = \idty - \Pi_-$, and $H_{\psi ++ } = \Pi_+ H_\psi \Pi_+$,
\begin{equation}
\begin{multlined}
    H_{2-+}(H_{W++})^{-1}H_{2+-} = \frac{1}{C'}\left[ \mathbf{1}_{AB}\otimes \left(\Pi_-\Gamma_i\Pi_+ (H_{\psi ++ })^{-1}\Pi_+\Gamma_i \Pi_- + \Pi_-\Gamma_j\Pi_+ (H_{\psi ++ })^{-1}\Pi_+\Gamma_j\Pi_-\right)\right.\\
\left.- P_A\otimes P_B\otimes \left(\Pi_-\Gamma_i\Pi_+ (H_{\psi ++ })^{-1}\Pi_+\Gamma_j\Pi_- + \Pi_-\Gamma_j\Pi_+ (H_{\psi ++ })^{-1}\Pi_+\Gamma_i\Pi_- \right) \right].
\end{multlined}
\end{equation}

Assuming that $C' \neq 0$, we can verify that the gadget works via \cref{second order sim}, as:
\begin{equation}
     \norm{T H_\textup{target}T^\dagger - H_{1--} + H_{2-+}(H_{W++})^{-1}H_{2+-} } = 0.
\end{equation}

We can now detail our argument.
The interaction strength of the simulation Hamiltonian is lower bounded by $\max(1/\sqrt{C'}, \mu_{\psi})$, where $\mu_\psi$ is the interaction strength of $H_\psi$.
Since $\tilde{H} = \Delta H_0 + H_1 + \sqrt{\Delta}H_2$ and we have assumed $\Delta\geq 1 $. 
Obtaining a large $C'$ requires that $H_\psi$ has large interaction strength. The crux of the argument is that $ \bra{\psi } \Gamma_i (H_{\psi ++ })^{-1} \Gamma_j \ket{\psi}$ is essentially a correlation function weighted by the inverse Hamiltonian, and systems with a large correlation functions need large interaction strengths.

Recall from the definition of $C'$,
\begin{align}
     C' &=  \bra{\psi } \Gamma_i (H_{\psi ++ })^{-1} \Gamma_j \ket{\psi} + \bra{\psi } \Gamma_j (H_{\psi ++ })^{-1} \Gamma_i \ket{\psi} \\
     & = 2 \textup{Re} ( \bra{\psi } \Gamma_i (H_{\psi ++ })^{-1} \Gamma_j \ket{\psi}) \\
     & \leq 2 |  \bra{\psi } \Gamma_i (H_{\psi ++ })^{-1} \Gamma_j \ket{\psi} | \\
     & = 2 \sqrt{ \bra{\psi } \Gamma_i (H_{\psi ++ })^{-1} \Gamma_j \ket{\psi}\bra{\psi } \Gamma_j (H_{\psi ++ })^{-1} \Gamma_i \ket{\psi} }.
\end{align}

Now we focus on the operator $(H_{\psi ++ })^{-1} \Gamma_j \ket{\psi}\bra{\psi } \Gamma_j (H_{\psi ++ })^{-1}$, since
\begin{equation}
   \norm{H_{\psi ++}}^{-1} \Pi_+ \leq (H_{\psi ++ })^{-1} \leq \Pi_+ ,
\end{equation}
where for two operators $A,B$, we say $A \geq B$ if $A-B \geq 0$, i.e. if $A-B$ is positive semi-definite.
Then
\begin{equation}
     \norm{H_{\psi ++}}^{-2}   \Pi_+ \Gamma_j \ket{\psi}\bra{\psi } \Gamma_j \Pi_+\leq (H_{\psi ++ })^{-1} \Gamma_j \ket{\psi}\bra{\psi } \Gamma_j (H_{\psi ++ })^{-1} \leq \Pi_+ \Gamma_j \ket{\psi}\bra{\psi } \Gamma_j \Pi_+
\end{equation}

We get
\begin{equation}
     \bra{\psi } \Gamma_i (H_{\psi ++ })^{-1} \Gamma_j \ket{\psi}\bra{\psi } \Gamma_j (H_{\psi ++ })^{-1} \Gamma_i \ket{\psi} \leq \bra{\psi } \Gamma_i \Pi_+\Gamma_j \ket{\psi}\bra{\psi } \Gamma_j \Pi_+\Gamma_i \ket{\psi}.
\end{equation}
Note that $\bra{\psi } \Gamma_i \Pi_+\Gamma_j \ket{\psi} = \bra{\psi} \Gamma_i \otimes \Gamma_j \ket{\psi} - \bra{\psi} \Gamma_i \dyad{\psi} \Gamma_j \ket{\psi}$ is indeed the usual correlation function. 
From Theorem 4 of \cite{Nachtergaele2010} and the fact that without loss of generality, we can take the distance between $\Gamma_i$ and $\Gamma_j$ to be $\Omega(n)$ -- we obtain\footnote{ For example taking $\Gamma_{\{i,j\}} = X_{\{i,j\}}$, the $\ket{W}$ state has $|\bra{W}X_i \Pi_+ X_j \ket{W}| = 2/n$, thus $\mu_W \in \widetilde{\Omega}(n)$}:
\begin{equation}
\abs{\bra{\psi } \Gamma_i \Pi_+\Gamma_j \ket{\psi}} \in O \left(  \norm{\Gamma_i}\norm{\Gamma_j} e^{- \Omega(n/\mu_\psi)}\right).
\end{equation}
This yields the desired result:
\begin{equation}
     C' \in O \left(  \norm{\Gamma_i}\norm{\Gamma_j} e^{- \Omega(n/\mu_\psi)}\right).
\end{equation}

Assuming that $\norm{\Gamma_i}, \norm{\Gamma_j}$ scale at most polynomially, the interaction strength of the simulation Hamiltonian can be lower bounded by 
\begin{equation}
  \max \left ( \mu_\psi, \Omega( e^{\Omega(n/\mu_\psi)} /\poly(n)) \right )  \leq \max \norm{\tilde{h}_i}.
\end{equation}
We can read off from the above equation that $\mu_\psi$ has to scale like $\in \widetilde{\Omega}(n)$ for $1/\sqrt{C'}$ to scale polynomially.\footnote{We note that while long correlation length cannot be achieved for a unique ground state, even a simple two-fold degenerate ground state (a $Z_iZ_{i+1}$ chain) contains states with long correlation length}

\section{Tighter estimation of the gap}
\label{sec:exact-norm}

We remind the reader of the following expression for $\Pi_{\mathcal{A} \cup \mathcal{B}} - \Pi_\mathcal{A} \Pi_\mathcal{B}$:
\begin{equation}
    \begin{multlined}
        \Pi_{\mathcal{A} \cup \mathcal{B}} - \Pi_\mathcal{A} \Pi_\mathcal{B} =  \left(\frac{\sqrt{a \bar{b}}}{m} - \sqrt{\frac{\bar{b}}{a}} \right) \ket{W}_\mathcal{A}\ket{0}_{\bar{\mathcal{A}}}\bra{W}_{\bar{\mathcal{B}}}\bra{0}_\mathcal{B} 
    + \left( \frac{\sqrt{ab}}{m} - \frac{l}{\sqrt{ab}}\right) \ket{W}_\mathcal{A}\ket{0}_{\bar{\mathcal{A}}}\bra{0}_{\bar{\mathcal{B}}}\bra{W}_\mathcal{B} \\
    + \left( \frac{\sqrt{\bar{a}b}}{m} - \sqrt{\frac{\bar{a}}{b}} \right)  \ket{0}_\mathcal{A} \ket{W}_{\bar{\mathcal{A}}} \bra{0}_{\bar{\mathcal{B}}} \bra{W}_\mathcal{B}
    +  \frac{\sqrt{\bar{a}\bar{b}}}{m} \ket{0}_\mathcal{A}\ket{W}_{\bar{\mathcal{A}}} \bra{W}_{\bar{\mathcal{B}}} \bra{0}_\mathcal{B} 
    -  \frac{\sqrt{\bar{a}\bar{b}}}{\sqrt{ab}} \ket{W}_\mathcal{A} \ket{W}_{\bar{\mathcal{A}}} \bra{W}_{\bar{\mathcal{B}}} \bra{W}_{\mathcal{B}}.
    \end{multlined}
\end{equation}

This can be rewriten as 
\begin{equation}
    \Pi_{\mathcal{A} \cup \mathcal{B}} - \Pi_\mathcal{A} \Pi_\mathcal{B} = \lambda_1 A_1 + \lambda_2 A_2 + \lambda_3 A_3 + \lambda_4 A_4 + \lambda_5 A_5
\end{equation}
where $\lambda_1 = \frac{\sqrt{a \bar{b}}}{m} - \sqrt{\frac{\bar{b}}{a}}$, $A_1 = \ket{W}_\mathcal{A}\ket{0}_{\bar{\mathcal{A}}}\bra{W}_{\bar{\mathcal{B}}}\bra{0}_\mathcal{B} $, and similarly for the rest.
Then, 
\begin{align}
    \norm{\Pi_{\mathcal{A} \cup \mathcal{B}} - \Pi_\mathcal{A} \Pi_\mathcal{B}}_\infty = & \norm{\sum_i \lambda_i A_i} = \max_{\ket{x}} \norm{\sum_i \lambda_i A_i \ket{x} } \\
    =  &\max_{\ket{x}} \sqrt{ \bra{x} \sum_i \lambda_i A_i^\dagger \sum_j \lambda_j A_j \ket{x}}.
\end{align}
Straightforward computation gives, 
\begin{equation}
     \sum_i \lambda_i A_i^\dagger \sum_j \lambda_j A_j  = A' \oplus A'',
\end{equation}
where 
\begin{equation}
\begin{multlined}
    A' = 
    (\lambda_2^2+\lambda_3^2)\dyad{0}_{\bar{\mathcal{B}}}\otimes \dyad{W}_{\mathcal{B}}
    + (\lambda_1\lambda_2 + \lambda_3\lambda_4) \ket{0}\bra{W}_{\bar{\mathcal{B}}}\otimes\ket{W}\bra{0}_{\mathcal{B}} \\
    + (\lambda_1 \lambda_2 + \lambda_3\lambda_4)\ket{W}\bra{0}_{\bar{\mathcal{B}}} \otimes \ket{0}\bra{W}_{\mathcal{B}} + (\lambda_1^2 + \lambda_4^2) \dyad{W}_{\bar{\mathcal{B}}}\otimes \dyad{0}_{\mathcal{B}},
    \end{multlined}
\end{equation}
and 
\begin{equation}
    A'' = \lambda_5^2 \dyad{W}_{\bar{\mathcal{B}}}\otimes \dyad{W}_{\mathcal{B}}.
\end{equation}

Since $A'$ is hermitian, it can be diagonalized, $A' = u_1 \dyad{u_1} + u_2\dyad{u_2}$.
Writing $p = \lambda_2^2 + \lambda_3^2$, $q = \lambda_1^2 + \lambda_4^2$, and $r = \lambda_1\lambda_2 + \lambda_3\lambda_4$, we get 
\begin{align}
  u_1 &= \frac{1}{2}(-\sqrt{p^2 - 2pq + q^2 + 4r^2} + p + q)  \\
  u_2 &= \frac{1}{2}(\sqrt{p^2 - 2pq + q^2 + 4r^2} + p + q).
\end{align}

Since $\sum_i \lambda_i A_i^\dagger \sum_j \lambda_j A_j $ can be diagonalized as ${u_1 \dyad{u_1} + u_2\dyad{u_2} + \lambda_5^2 \dyad{W}_{\bar{\mathcal{B}}}\otimes \dyad{W}_{\mathcal{B}}}$, then 
\begin{equation}
    \max_{\ket{x}} \bra{x} \sum_i \lambda_i A_i^\dagger \sum_j \lambda_j A_j \ket{x} = \max(u_1, u_2, \lambda_5^2).
\end{equation}

We conclude that $\norm{\Pi_{\mathcal{A} \cup \mathcal{B}} - \Pi_\mathcal{A} \Pi_\mathcal{B}}_\infty = \max(\sqrt{u_1}, \sqrt{u_2}, \abs{\lambda_5})$.
Using this expression, and plotting the exponent $\log_{\frac{1+\gamma}{2}}(\epsilon^{-1})$ as a function of $\gamma$, we find that $\lambda_\Lambda \in \Omega(n^{-6.13})$ for $\gamma = 0.552$.
\end{appendix}

\printbibliography
\end{document}